\documentclass[lettersize,journal]{IEEEtran}
\usepackage{array}
\usepackage[caption=false,font=normalsize,labelfont=sf,textfont=sf]{subfig}
\usepackage{textcomp}
\usepackage{stfloats}
\usepackage{verbatim}
\usepackage{dsfont}

\usepackage{graphicx, amssymb, amsmath, amsthm,amsfonts}
\usepackage{extarrows}

\usepackage{algorithm}
\usepackage{algorithmicx}
\usepackage{algpseudocode}

\usepackage{multirow} 
\usepackage{color}
\usepackage{cite}
\usepackage{booktabs}
\usepackage[numbers,sort&compress]{natbib}
\usepackage{url}
\newtheorem{definition}{Definition}

\begin{document}

\title{Placing Timely Refreshing Services \\ at the Network Edge}

\author{Xishuo~Li,~\IEEEmembership{Student~Member,~IEEE,}
				Shan~Zhang,~\IEEEmembership{Member,~IEEE,}	
				Hongbin~Luo,~\IEEEmembership{Member,~IEEE,}
				Xiao~Ma,~\IEEEmembership{Member,~IEEE,}
                Junyi~He,~\IEEEmembership{Student~Member,~IEEE,}


\thanks{This work was supported in part by the National Key R\&D Program of China under Grant 2022YFB4501000, in part by the Nature Science Foundation of China under Grant 62271019, 62225201, and in part by National Engineering Research Center of Advanced Network Technologies. (\textit{Corresponding author: Shan Zhang})}

\thanks{X. Li, S. Zhang, and H. Luo are with School of Computer Science and Engineering, Beihang University, Beijing 100191, China. (e-mail: lixishuo@buaa.edu.cn; zhangshan18@buaa.edu.cn; luohb@buaa.edu.cn)}

\thanks{X. Ma is with the State Key Laboratory of Networking and Switching Technology, Beijing University of Posts and Telecommunications, Beijing 100876, China (e-mail: maxiao18@bupt.edu.cn)}

\thanks{J. He is with the School of Science and Engineering, the Chinese University of Hong Kong, Shenzhen 518172, China (e-mail: junyihe@link.cuhk.edu.cn)}

\thanks{Copyright (c) 20xx IEEE. Personal use of this material is permitted. However, permission to use this material for any other purposes must be obtained from the IEEE by sending a request to pubs-permissions@ieee.org.}
}

\markboth{Journal of \LaTeX\ Class Files,~Vol.~14, No.~8, August~2021}%
{Shell \MakeLowercase{\textit{et al.}}: A Sample Article Using IEEEtran.cls for IEEE Journals}


\maketitle

\begin{abstract}
Accommodating services at the network edge is favorable for time-sensitive applications. However, maintaining service usability is resource-consuming in terms of pulling service images to the edge, synchronizing databases of service containers, and hot updates of service modules. Accordingly, it is critical to determine which service to place based on the received user requests and service refreshing (maintaining) cost, which is usually neglected in existing studies. In this work, we study how to cooperatively place timely refreshing services and offload user requests among edge servers to minimize the backhaul transmission costs. We formulate an integer non-linear programming problem and prove its NP-hardness. This problem is highly non-tractable due to the complex spatial-and-temporal coupling effect among service placement, offloading, and refreshing costs. We first decouple the problem in the temporal domain by transforming it into a Markov shortest-path problem. We then propose a light-weighted Discounted Value Approximation (DVA) method, which further decouples the problem in the spatial domain by estimating the offloading costs among edge servers. The worst performance of DVA is proved to be bounded. 5G service placement testbed experiments and real-trace simulations show that DVA reduces the total transmission cost by up to 59.1\% compared with the state-of-the-art baselines.
\end{abstract}


\section{Introduction}
\label{sec:introduction}
\IEEEPARstart{B}{y} moving services from cloud datacenters to the network edge, mobile edge computing (MEC) can provision service containers on edge servers to process user requests in proximity and with low latency, and thus support time-critical applications \cite{DBLP:journals/comsur/AbkenarRIMCWJR22}. The development of MEC largely benefits from the recent progress of serverless computing, which enables the fast container provisioning process, i.e., creating and starting the service container instance(s) by loading the service image files in 50-500ms\cite{DBLP:journals/wc/XieTQZYH21,DBLP:conf/usenix/AkkusCRSSBAH18}. However, pulling the image files from a registry in the cloud to an edge server via the constrained backhaul network will lead to a significant transmission delay, e.g., 5-10s \cite{DBLP:conf/nsdi/ChenLGL22}. Hence, instead of pulling image files when requests arrived, existing MEC systems usually cache appropriate service images on edge servers in advance \cite{KubeEdge}.

\begin{figure}[tbp]
\centering
\includegraphics[width=\linewidth]{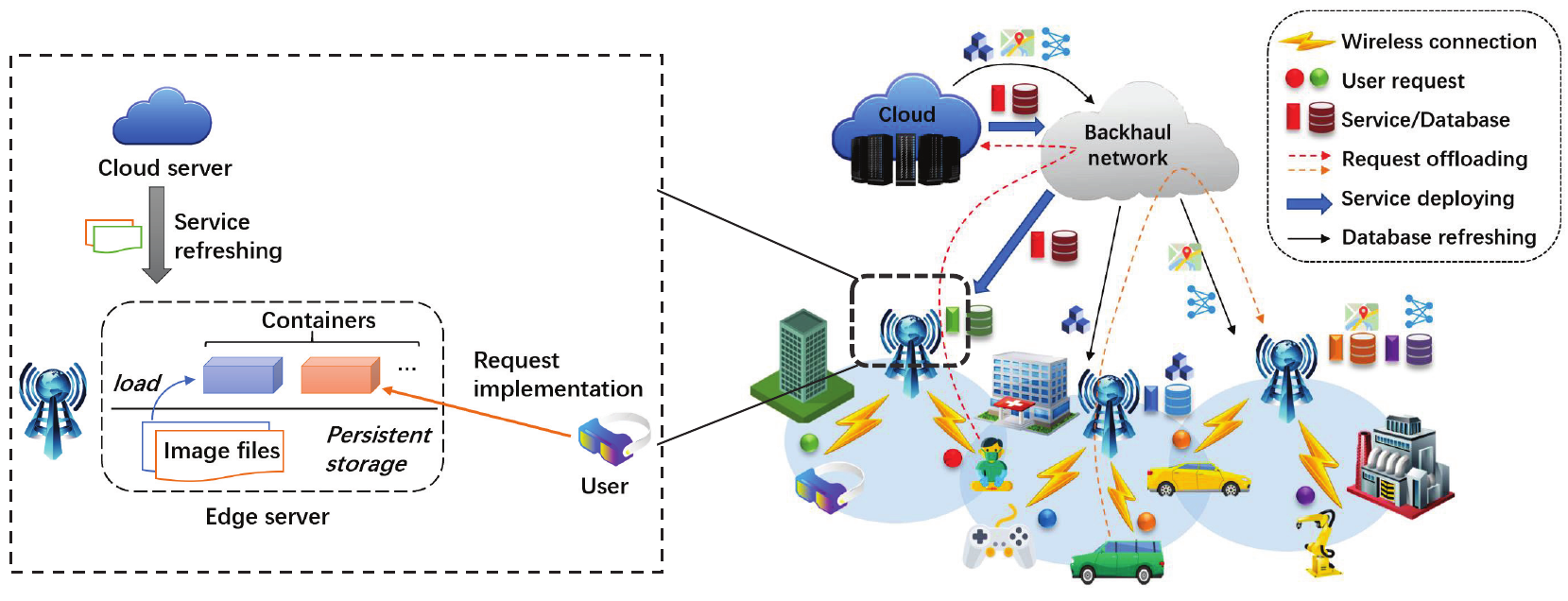}
\caption{An illustration of the MEC system, the left part presents the service placement and refreshing mechanism.}
\label{system}
\end{figure}

After caching images, the services placed on edge servers should be \textit{refreshed} timely to maintain their usability and functionality. On the one hand, the updating of the service databases in the cloud server should be synchronized with the service container instances provisioned at the network edge, e.g., traffic information, stock prices, etc \cite{bestavros2012real,Redis}. For example, autonomous driving services require to be updated with the latest high-definition maps of nearby streets to support route planning \cite{liu2020high}. On the other hand, the image files cached on edge servers also should be refreshed for hot updates of software modules, security patches, container re-configurations, and so on \cite{DBLP:conf/nsdi/ChenLGL22}. For instance, AI-based healthcare services require to be updated with the latest and most accurate neural network model to detect diseases \cite{rieke2020future}. Such service refreshing traffic brings considerable transmission pressure on the backhaul network \cite{DBLP:journals/dpd/ShankerMS08,DBLP:journals/csur/NashatA18}. For example, transmitting the latest 3D point cloud data (traffic information) of nearby streets requires a bandwidth exceeding 100 Mbps \cite{DBLP:conf/vrml/CaoPZ19,DBLP:journals/comsur/WangLK19}. Generally, accommodating services at MEC is extremely resource-consuming due to frequently pulling service images and refreshing placed services.

Given that base stations and edge servers are densely deployed in 5G and beyond networks \cite{DBLP:journals/jsac/ShafiMSHZSTBW17}, it is a critical problem of choosing appropriate services to place and refresh on each edge server to serve user requests as more as possible, considering the limited MEC capacity and backhaul bandwidth. This problem of dynamically selecting the appropriate services from a huge set of candidates \footnote{For example, there are three million applications available on Google Play by 2022 \cite{num_in_googleplay}.} is also known as the \textit{service placement problem} or \textit{service provisioning problem} \cite{DBLP:journals/csur/Ait-SalahtDL20,DBLP:journals/comsur/SonkolyCSNT21}.

Previous studies have made great progress in service placement problems, considering system dynamics in terms of the user requests \cite{DBLP:journals/jsac/Karimzadeh-Farshbafan20,DBLP:conf/infocom/FloresTT20,DBLP:journals/cn/GolkarifardCMM21}, user mobility \cite{DBLP:journals/tmc/ZhaoLTTQ21,DBLP:journals/tpds/NingDWWHGQHK21}, and service dependency (e.g., correlated service components) \cite{DBLP:journals/comsur/SonkolyCSNT21,DBLP:journals/csur/Ait-SalahtDL20}. Nevertheless, these studies ignored the refreshing demand of services, which will lead to invalid services (users cannot be served) or overburdened backhaul (degraded network performance) in practical systems. Thus, we revisit the service placement problem, considering both the gain (serving requests at edge) and the pain (refreshing demand) of services, to minimize the overall backhaul transmission cost. Note that our work focuses on caching service image files on edge servers, which is the prerequisite of processing requests in MEC. This work is compatible with the fine-grained optimization studies in user association\cite{DBLP:journals/tnsm/BehraveshHCR21,tnse9767553}, computation offloading\cite{DBLP:journals/tvt/HuangLMZ21,DBLP:journals/corr/abs-2210-17025,DBLP:journals/iotj/QianZWWL23}, container lifecycle management \cite{DBLP:journals/spe/WuTFHZJYC22} and serverless computing \cite{DBLP:journals/wc/XieTQZYH21}, which usually assume the service images and computing environments are ready.

We consider a typical MEC system, as shown in Fig. \ref{system}. The service placement problem is explored in a cooperated multi-MEC scenario, wherein the placed services are refreshed periodically based on the corresponding timeliness requirements. Each edge server decides which services to place according to the upcoming user requests and the timeliness status of placed services. User requests will be served directly by the home edge server or offloaded to neighboring edge servers or the cloud, which will lead to different backhaul pressure. This \underline{T}imely \underline{R}efreshing \underline{S}ervice \underline{P}lacement problem (TRSP) is formulated as an NP-hard integer non-linear programming problem.

The main difficulty of TRSP is its complex spatial-and-temporal coupling property. Firstly, in the temporal domain, the service refreshing couples with the service placement decisions across multiple time slots because the timeliness of service decays linearly with time. Moreover, in the spatial domain, the refreshing and placement of services also couple with the request offloading decisions among different edge servers. As an non-Markov problem, TRSP can hardly be solved by existing methods which usually assume the decisions in different time slots are independent \cite{DBLP:journals/comsur/SonkolyCSNT21}. Besides, it is also computationally intractable to utilize exhaustive search to solve TRSP, whose time complexity grows exponentially with the number of edge servers, available services, and time slots.

To address the TRSP problem, we add the refreshing timer information into the system state space to decouple the service placement and offloading decisions in the temporal domain, whereby TRSP becomes a Markov problem and can be solved in a slot-by-slot manner. In this way, the TRSP problem equals a shortest path problem, and then we propose a dynamic programming-based centralized algorithm to obtain the optimal solution. The time complexity of the centralized method grows linearly (resp. exponentially) with the number of time slots (resp. the number of edge servers and available services).

To reduce the computational complexity, we design a light-weighted \underline{D}iscounted \underline{V}alue \underline{A}pproximation algorithm (DVA) for practical large-scale scenarios, which further decouples the transformed problem spatially. Specifically, in the spatial domain, DVA will estimate the offloading cost among edge servers based on the state of each edge server's neighbors. Then, DVA will predict the future backhaul transmission cost of current strategy to make decisions. The approximation ratio (gap between DVA and the optimal results) of DVA is analyzed theoretically, which is small when the popularity of services keeps stable. 5G service placement testbed experiments and extensive real-trace simulations show that DVA is close-to-optimal and reduces the backhaul transmission cost by up to 59.1\% compared with the state-of-the-art baselines. The results also indicate that setting larger storage space on edge servers may even increase the transmission cost since the service refreshing traffic may exceed the gain of serving users at network edge.

The main contributions of this work are as follows:
\begin{enumerate}
\item We study the problem of how to cooperatively place services and offload requests at the network edge, considering the practical refreshing demands of services. Aiming at minimizing the backhaul transmission cost, we formalize TRSP as an integer non-linear programming problem and prove its NP-hardness. The insight behind our problem is balancing the gain of serving more requests at the edge and the pain of increased backhaul pressure to maintain service usability.

\item TRSP is highly non-tractable due to the complex spatial-and-temporal coupling effect among service placement, offloading, and refreshing costs. We first decouple the problem in the temporal domain by modifying the system state space, whereby TRSP is transformed into a Markov shortest path problem. Then a dynamic programming-based algorithm with exponential complexity is proposed to obtain the optimal solution in a slot-by-slot manner.

\item To reduce the computation complexity, we propose a light-weighted algorithm, DVA, for large-scale scenarios. DVA further decouples TRSP in the spatial domain by approximating the future transmission cost on each edge server. The worst performance of DVA is proved to be bounded. 5G service placement testbed experiments and extensive real-trace simulations demonstrate that DVA is close-to-optimal and can reduce the backhaul transmission cost by up to 59.1\% compared with the state-of-the-art baselines.
\end{enumerate}

The rest of our paper is organized as follows. Section \ref{sec_review} reviews the related work. Next, section \ref{sec_problem_formu} describes the system model. Section \ref{pfa} formulates the TRSP problem and analyzes its complexity. Then, in section \ref{sec_alg_solutions}, both the centralized and distributed algorithms are proposed to solve the TRSP problem. Section \ref{sec_simulation} presents the experimental results. In the last section \ref{sec_conclusions}, we conclude this study and discuss future work.

\renewcommand\arraystretch{1.2}
\begin{table*}[]
\caption{Classifications of Representative Service Placement Studies}
\label{tab:my-table}
\centering
\begin{tabular}{|c|cccc|cc|}
\hline
\multirow{3}{*}{\textbf{Optimization goals}} &
  \multicolumn{4}{c|}{\textbf{MEC scenario}} &
  \multicolumn{2}{c|}{\textbf{Methodology}} \\ \cline{2-7} 
 &
  \multicolumn{1}{c|}{\multirow{2}{*}{\textit{\begin{tabular}[c]{@{}c@{}}Static\\ scenario\end{tabular}}}} &
  \multicolumn{3}{c|}{\textit{Dynamic scenario}} &
  \multicolumn{1}{c|}{\multirow{2}{*}{\textit{\begin{tabular}[c]{@{}c@{}}Offline\\ methods\end{tabular}}}} &
  \multirow{2}{*}{\textit{\begin{tabular}[c]{@{}c@{}}Online\\ methods\end{tabular}}} \\ \cline{3-5}
 &
  \multicolumn{1}{c|}{} &
  \multicolumn{1}{c|}{\textit{\begin{tabular}[c]{@{}c@{}}Dynamic network\\ infrastructure\end{tabular}}} &
  \multicolumn{1}{c|}{\textit{\begin{tabular}[c]{@{}c@{}}Dynamic\\ applications\end{tabular}}} &
  \textit{\begin{tabular}[c]{@{}c@{}}Mobile\\ users\end{tabular}} &
  \multicolumn{1}{c|}{} &
   \\ \hline
Service latency &
  \multicolumn{1}{c|}{\cite{DBLP:journals/ett/RaghavendraCG21,DBLP:journals/iotj/FangM21,DBLP:journals/tnsm/LiuGLY21,DBLP:journals/tpds/ChenZJQXGL22,DBLP:journals/ett/BhamareSEJG18}} &
  \multicolumn{1}{c|}{\cite{DBLP:journals/iotj/SarkarAKK22}} &
  \multicolumn{1}{c|}{\cite{DBLP:conf/vtc/TalpurG21,DBLP:journals/iotj/ZhangYPC20}} &
  \cite{DBLP:journals/tmc/ZhaoLTTQ21,DBLP:journals/tpds/NingDWWHGQHK21,DBLP:conf/vtc/TalpurG21} &
  \multicolumn{1}{c|}{\cite{DBLP:journals/iotj/SarkarAKK22,DBLP:journals/ett/RaghavendraCG21,DBLP:journals/iotj/FangM21,DBLP:journals/tnsm/LiuGLY21,DBLP:journals/ett/BhamareSEJG18}} &
  \cite{DBLP:journals/tmc/ZhaoLTTQ21,DBLP:journals/tpds/NingDWWHGQHK21,DBLP:conf/vtc/TalpurG21,DBLP:journals/iotj/ZhangYPC20,DBLP:journals/tpds/ChenZJQXGL22} \\ \hline
Energy consumption &
  \multicolumn{1}{c|}{\cite{DBLP:journals/ett/RaghavendraCG21,DBLP:journals/iotj/FangM21,DBLP:journals/tnsm/LiuGLY21,DBLP:journals/tsc/Duong-BaTNB21}} &
  \multicolumn{1}{c|}{\cite{DBLP:conf/ispa/MaaouiaFCJ18}} &
  \multicolumn{1}{c|}{\cite{DBLP:journals/tgcn/ChenSHTF21}} &
  \cite{DBLP:journals/tmc/ZhaoLTTQ21} &
  \multicolumn{1}{c|}{\cite{DBLP:journals/ett/RaghavendraCG21,DBLP:journals/iotj/FangM21,DBLP:journals/tnsm/LiuGLY21,DBLP:journals/ett/BhamareSEJG18,DBLP:journals/tsc/Duong-BaTNB21}} &
  \cite{DBLP:journals/tmc/ZhaoLTTQ21,DBLP:journals/tgcn/ChenSHTF21,DBLP:conf/ispa/MaaouiaFCJ18} \\ \hline
Resource utilization &
  \multicolumn{1}{c|}{\cite{DBLP:journals/tsc/Duong-BaTNB21,DBLP:conf/icc/AroraK21}} &
  \multicolumn{1}{c|}{\cite{DBLP:conf/networking/MaoWZX20}} &
  \multicolumn{1}{c|}{\cite{DBLP:conf/vtc/TalpurG21,DBLP:journals/tgcn/ChenSHTF21}} &
  \cite{DBLP:conf/vtc/TalpurG21} &
  \multicolumn{1}{c|}{\cite{DBLP:journals/tsc/Duong-BaTNB21,DBLP:conf/icc/AroraK21}} &
  \cite{DBLP:journals/tgcn/ChenSHTF21,DBLP:conf/vtc/TalpurG21,DBLP:conf/networking/MaoWZX20} \\ \hline
Placement cost &
  \multicolumn{1}{c|}{\cite{DBLP:journals/iotj/NguyenNTB22,DBLP:journals/tpds/ChenZJQXGL22,DBLP:journals/ett/BhamareSEJG18,DBLP:conf/netsoft/KhoshkholghiTBK19}} &
  \multicolumn{1}{c|}{\cite{DBLP:conf/wcnc/LiLCX21}} &
  \multicolumn{1}{c|}{\cite{DBLP:journals/jsac/Karimzadeh-Farshbafan20,DBLP:conf/infocom/FloresTT20}} &
  \cite{DBLP:conf/wcnc/LiLCX21,DBLP:journals/tnsm/BehraveshHCR21} &
  \multicolumn{1}{c|}{\cite{DBLP:conf/infocom/FloresTT20,DBLP:journals/tnsm/BehraveshHCR21,DBLP:journals/ett/BhamareSEJG18,DBLP:conf/netsoft/KhoshkholghiTBK19}} &
  \cite{DBLP:conf/wcnc/LiLCX21,DBLP:journals/jsac/Karimzadeh-Farshbafan20,DBLP:journals/iotj/NguyenNTB22,DBLP:journals/tpds/ChenZJQXGL22} \\ \hline
Others (e.g., throughput) &
  \multicolumn{1}{c|}{\cite{DBLP:journals/iotj/NguyenNTB22}} &
  \multicolumn{1}{c|}{\cite{DBLP:journals/fgcs/XieWD20,DBLP:conf/globecom/DalgkitsisM0KV20,DBLP:conf/networking/MaoWZX20}} &
  \multicolumn{1}{c|}{\cite{DBLP:journals/cn/GolkarifardCMM21,DBLP:journals/cn/YuanXYLCTGPW20}} &
  \cite{DBLP:journals/tnsm/BehraveshHCR21} &
  \multicolumn{1}{c|}{\cite{DBLP:journals/cn/GolkarifardCMM21,DBLP:journals/tnsm/BehraveshHCR21}} &
  \cite{DBLP:journals/cn/GolkarifardCMM21,DBLP:journals/cn/YuanXYLCTGPW20,DBLP:journals/fgcs/XieWD20,DBLP:conf/globecom/DalgkitsisM0KV20,DBLP:journals/iotj/NguyenNTB22,DBLP:conf/networking/MaoWZX20} \\ \hline
\end{tabular}
\end{table*}

\section{Literature Review}
\label{sec_review}

Mobile edge computing can provide low-latency service to end-users and relieve pressure on the backhaul network and cloud servers. A body of works have been conducted on service placement problems with resource constraints (e.g., caching, communication, and computation) to optimize various goals (e.g., service latency and resource utilization. We summarize some representative studies in table \ref{tab:my-table}. For more details, please refer to the survey papers \cite{DBLP:journals/csur/Ait-SalahtDL20,DBLP:journals/comsur/SonkolyCSNT21}. Generally, the existing works can be classified into static and dynamic categories based on the MEC scenarios, or offline and online categories according to whether the decisions are made based on future information. For convenience, we introduce these studies according to their scenarios.

The static scenario indicates that the MEC system is unchanged with time (e.g., without mobile users, breakdown of links/nodes). For instance, Chen \textit{et al.} designed a local-search based algorithm to allocate users and place services in edge environment, aiming at minimizing the service deployment and resource consumption cost \cite{DBLP:journals/tpds/ChenZJQXGL22}. Besides, Bhamare \textit{et al.} proposed an affinity-based scheduling method to reduce the overall turnaround time and deployment costs for scheduling microservices across multiple clouds \cite{DBLP:journals/ett/BhamareSEJG18,DBLP:conf/netsoft/KhoshkholghiTBK19, Gupta:2016:2215-0811:9}.

The dynamic scenario includes three aspects: dynamic network infrastructure, dynamic applications, and mobile end users. The dynamic network infrastructure corresponds to the join, leave, or breakdown of network entities and links. As an example, Sarkar \textit{et al.} proposed a deadline-aware strategy to dynamically select the suitable computing devices for IoT applications during failure in federated fog clusters \cite{DBLP:journals/iotj/SarkarAKK22}. Similarly, Li \textit{et al.} presented an online service placement approach to enable device-to-device offloading between mobile (dynamic) devices, aiming to minimize the computation, communication, and migration cost \cite{DBLP:conf/wcnc/LiLCX21}.

The dynamic applications represent that the applications change with time, e.g., the resource requirements are not fixed, and the application components may break down. For example, to handle the dynamic arrivals and departures of network function virtualization services,  Karimzadeh-Farshbafan \textit{et al.} designed a Viterbi-based reliability-aware placement method to minimize the placement cost and maximize the number of admitted services \cite{DBLP:journals/jsac/Karimzadeh-Farshbafan20}. Besides, Flores \textit{et al.} proposed a policy-aware virtual machine placement strategy to accommodate the dynamic traffic loads in data centers and minimize the communication cost \cite{DBLP:conf/infocom/FloresTT20}.

The mobility of end-users is a complex problem in service placement studies. One primary consideration is to decide whether to migrate services to places near the mobile users. As an example, Behravesh \textit{et al.} presented a heuristic algorithm to dynamically associate mobile users in edge networks and minimize the interruption of service function chain tasks \cite{DBLP:journals/tnsm/BehraveshHCR21}. Besides, Ning \textit{et al.} presented a Lyapunov optimization method to handle the long-term service migration problem in highly dynamic mobile networks with the objective of maximizing the system utility \cite{DBLP:journals/tpds/NingDWWHGQHK21}.

Compared with previous studies, our work takes the practical service refreshing demands into consideration. Aiming to minimize the backhaul transmission cost, we intend to balance the gain of serving more requests at the edge and the pain of increased backhaul pressure to maintain service usability. Due to the complex spatial-and-temporal coupling effect among service placement, refreshing, and offloading costs, the existing approaches can hardly be applied to solve the non-Markov problem. Thus, a novel light-weighted method should be proposed.

\section{System Model}
\label{sec_problem_formu}
In this section, we describe the details of the MEC system, including the service placement and refreshing mechanism. Then, the network model, the service placement model, and the service refreshing model are also formulated. The key notations in the system model are summarized in Table \ref{tab:key notations}.

\subsection{Service Placement and Refreshing Mechanism}
\label{TRSPm}

The considered service placement and refreshing mechanism is illustrated in Fig. \ref{system}. When a service is placed at the network edge, the image file of the service is cached on the edge server accordingly. Upon the arrival of user requests, one or more container instance(s) are created to serve the requests by loading the image file \cite{KubeEdge}. In addition, the container will be dynamically destroyed to release the occupied resources, e.g., memory, CPU, when the processing of the request finishes \cite{DBLP:conf/usenix/AkkusCRSSBAH18,DBLP:conf/icdcs/CastroIMS17} Through dynamically creating and destroying containers, the resources can be better utilized at the edge servers \cite{DBLP:journals/wc/XieTQZYH21}.

The databases and service modules of service container instances should be refreshed with the latest information synchronized from the cloud server timely \cite{DBLP:journals/csur/NashatA18}. For example, the financial trading information is usually updated every second \cite{DBLP:journals/access/MonratSA19}; the traffic information required by autonomous driving services should be updated every minute \cite{liu2020high}; the analysis models in healthcare services are updated in every hour \cite{rieke2020future}. Without loss of generality, the considered services are refreshed periodically. By setting appropriate refreshing periods, it can be ensured that the placed services will be requested shortly once refreshed with high probability. As for aperiodically refreshed services, the refreshing interval can be normalized into fixed length, and the normalization error can be migrated in the long term.

To maintain the flexibility of request offloading, the \textit{services} discussed in this paper are stateless services, such that different requests are independent and the container instances of the same service created on different edge servers are identical. Besides, the considered services are monolithic services, i.e., all service components are packed into one image file. The problem of placing dependable services (e.g., service function chain) is much more complicated due to the coupling property among distributed service components \cite{DBLP:journals/comsur/SonkolyCSNT21}, which is left as our future work.

\subsection{Network Model}
As shown in Fig. \ref{system}, we consider a system that consists of multiple edge servers \(\mathbb{N}=\left \{ 1,\cdots ,n \right \} \) and multiple end-users requesting a set of services \(\mathbb{S}=\left \{ 1,\cdots ,s \right \} \). We denote the image file size of service \(i\) by \(r_{i}\) and the storage capacity of edge server \(j\) by \(R_{j}\). Thus, the edge servers can only store a few number of image files (place few services). Besides, the cloud server (indexed by 0) runs all services and can process any received requests.

The discrete-time model is adopted, and the timeline consists of time slots \(\mathbb{T}=\left \{ 1,\cdots ,T \right \}\). The length of one time slot is one minute, such that the placed services can be ready to serve users. Besides, the minute-level decision making process on edge servers can handle the request dynamics of mobile users properly \cite{DBLP:conf/infocom/ZhouWTZ22}. We use \(\lambda _{i,j}^{(t)}\) to represent the (predicted) number of received requests for service \(i\) at edge server \(j\) in slot \(t\), which depends on the service popularity, user mobility, and status of the surrounding environments. In practical systems, such information is usually obtained via time series prediction methods \cite{lim2021time,tan2021high}, in which deep learning-based approaches are commonly utilized to establish the user request model/pattern based on the history request information.

\subsection{Service Placement Model}
Denote by $\alpha_i$ the volume of data transmitted via backhaul when placing service \(i\) at edge servers. The placed services should be refreshed timely to provide a satisfying user experience. We denote \(\beta_i\) as the volume of data downloaded via backhaul for refreshing service \(i\).

In the system, users will be served directly if the requested service is placed at the home edge server (the edge server that a user directly connects via wireless). Otherwise, the requests will be offloaded to neighboring edge servers or the cloud. We denote \(d_i\) as the volume of data transmitted via the backhaul for offloading one request of service \(i\). Generally, transmitting data between edge and cloud needs to go through the backbone network, leading to a higher backhaul pressure than transmission between edge servers \cite{DBLP:journals/jsac/ShafiMSHZSTBW17}. The traffic coefficient between edge server \(j\) and edge server \(k\), i.e., the cost of transmitting one unit (GB) data from \(j\) to \(k\), is denoted by \(\gamma _{j,k}\ge 0\). For convenience, we set the traffic coefficient between edge servers and the cloud as \(\gamma _{j,0}= 1, \forall j\in\left \{ 1,\cdots ,n \right \} \). Therefore, the transmission cost of offloading all requests for service \(i\) from edge server \(j\) to \(k\) in slot \(t\) is \(\gamma_{j,k}\lambda _{i,j}^{(t)}d_i\).

Finally, we denote by \(x_{i,j}^{(t)}\in\left \{ 0,1 \right \} \) a binary indicator showing whether service \(i\) is placed at edge server \(j\) in slot \(t\). Let \(y_{i,j,k}^{(t)}\in\left \{ 0,1 \right \} \) represent whether to offload service \(i\)'s requests from edge server \(j\) to edge server \(k\) (or the cloud server) in slot \(t\).

\renewcommand\arraystretch{1.3}
\begin{table}[]
\centering
\caption{Key Notations}
\label{tab:key notations}
\begin{tabular}{@{}cp{6.5cm}@{}}
\toprule
Notation           & Description                                                        \\ \midrule
\(r_{i}\)          & image file size of service \(i\)                                              \\
\(R_{j}\)          & storage capacity of edge server \(j\)                                          \\
\(\alpha_i\)       & volume of data transmitted via backhaul for placing service \(i\)                    \\
\(\beta _{i}\)     & volume of data transmitted via backhaul for refreshing service \(i\) \\
\(d_{i}\)          & volume of data transmitted via backhaul for offloading one request of service \(i\) \\
\(\gamma _{j,k}\)  & traffic coefficient between edge server \(j\) and \(k\)               \\
\(\lambda _{i,j}^{(t)}\) & the number of received requests for service \(i\) at edge server \(j\) in slot \(t\)                                   \\
\(l _{i,j}^{(t)}\) & remaining lifetime of service \(i\) at edge server \(j\) in slot \(t\)      \\
\(x_{i,j}^{(t)}\)        & binary indicator showing whether service \(i\) is placed at edge server \(j\) in slot \(t\)                     \\
\(y_{i,j,k}^{(t)}\)      & binary indicator showing whether to offload requests for service \(i\) from edge server \(j\) to \(k\) in slot \(t\)\\
\(Q^{(t)}\)        & system state in slot \(t\)                                         \\
\(\theta\)         & temporal approximation factor for predicting the future transmission cost                       \\
\(\delta _{i,j}\)  & spatial approximation factor of service \(i\) at edge server \(j\) for predicting the transmission cost of offloading service \(i\)'s requests              \\
\bottomrule
\end{tabular}
\end{table}

\subsection{Service Refreshing Model}
In the system, service \(i\) is refreshed every \(LF_i\) slots to maintain its usability. Namely, the refreshing interval is fixed for each specific service. Denote by \(l _{i,j}^{(t)}\) the remaining lifetime (i.e., remained time slots before next refreshing) of service \(i\) at edge server \(j\) in slot \(t\), which decreases linearly with time. When \(l _{i,j}^{(t)}\) decreases to zero, service \(i\) at edge server \(j\) must be refreshed. Without loss of generality, \(l_{i,j}^{(t)}\) is given by:

\begin{equation}
\label{remaining_lifetime}
l_{i,j}^{(t)}=
\left\{\begin{matrix}
l_{i,j}^{(t-1)}-1, & \mathrm{if}\ x_{i,j}^{(t-1)}=1,\ x_{i,j}^{(t)}=1\ \mathrm{and}\ l_{i,j}^{(t-1)}\ge1\\ 
LF_{i}, & \mathrm{otherwise},
\end{matrix}\right.
\end{equation}
where \( x_{i,j}^{(0)}=0, \forall i\in \mathbb{S} ,\forall j\in \mathbb{N}\) for initialization. We utilize a function \(\varphi \left ( l_{i,j}^{(t)} \right )\) to represent whether service \(i\) at edge server \(j\) should be refreshed in slot \(t\) based on the remaining lifetime \(l_{i,j}^{(t)}\), given by:

\begin{equation}
\varphi \left ( l_{i,j}^{(t)} \right )=\left\{\begin{matrix}
1, & \mathrm{if}\  l_{i,j}^{(t)}=0,\\ 
 0,& \mathrm{otherwise}.
\end{matrix}\right.
\end{equation}

\section{Problem Formulation and Analysis}
\label{pfa}
In this section, we formulate the timely refreshing service placement problem and analyze its hardness.

\subsection{Problem Formulation}

The backhaul transmission cost of the edge server \(j\) includes three parts: (1) \(D_{j}^{\mathrm{pla} }\) : the cost of placing services at edge server \(j\), (2) \(D_{j}^{\mathrm{ref} }\) : the cost of refreshing services at edge server \(j\), and (3) \(D_{j}^{\mathrm{off} }\): the cost of offloading user requests from edge server \(j\). Accordingly, The backhaul transmission cost of the edge server \(j\) in slot \(t\) is given by

\begin{equation}
\label{trans_cost_j}
\begin{split}
D_{j}^{(t)}&=D_{j}^{\mathrm{pla}}+D_{j}^{\mathrm{ref} }+D_{j}^{\mathrm{off}}\\
&=\sum_{i=1}^{s} \left ( 1-x_{i,j}^{(t-1)} \right ) x_{i,j}^{(t)}\alpha _{i}
+\sum_{i=1}^{s} \varphi \left ( l_{i,j}^{(t)} \right ) \beta _{i}\\&
+\sum_{i=1}^{s} \sum_{k=0}^{n} \left ( 1-x_{i,j}^{(t)} \right ) \lambda _{i,j}^{(t)}d_{i}\cdot y_{i,j,k}^{(t)}\gamma _{j,k} .
\end{split}
\end{equation}

The TRSP problem is formulated as follows:

\begin{align}
\mathrm{(P1)}\quad \min_{x,y} & \ D_{total} = \sum_{t = 1}^{T} \sum_{j = 1}^{n} D_{j}^{(t)} \nonumber \\
\mathrm{s.t.} 
& \sum_{i=1}^{s} x_{i,j}^{(t)}r_{i} \le R_{j}, \forall j\in \mathbb{N} ,\forall t\in \mathbb{T},\\
& \sum_{k = 0}^{n} y_{i,j,k}^{(t)} = 1, \forall i\in \mathbb{S} ,\forall j\in \mathbb{N} ,\forall t\in \mathbb{T},\\
& y_{i,j,k}^{(t)}\le x_{i,k}^{(t)}, \nonumber \\ & \forall i\in \mathbb{S} ,\forall j\in \mathbb{N} , \forall k\in \left \{ 0,\cdots ,n \right \}   ,\forall t\in \mathbb{T},\\
& x_{i,j}^{(t)}, y_{i,j,k}^{(t)} \in \left \{ 0,1 \right \} , \nonumber \\ &  \forall i\in \mathbb{S} ,\forall j\in \mathbb{N} ,\forall k\in \left \{ 0,\cdots ,n \right \},\forall t\in \mathbb{T},
\end{align}
where \(D_{total}\) represents the total backhaul transmission cost of all edge servers in all slots. Constraint (4) indicates that the total image file size of placed services cannot exceed the storage limit. The constraint (5) ensures that all user requests will be processed but not discarded. Constraint (6) represents that user requests can only be offloaded to the server that has placed the corresponding service. The last constraint (7) is the binary constraint of the decision variables. The TRSP problem is a non-linear integer optimization problem, which is notoriously hard to solve.

According to the system model, optimizing the backhaul transmission cost leads to serving requests at edge servers as more as possible. Accordingly, minimizing the backhaul transmission cost is also consistent with reducing the delay of user requests, which is the objective of many existing researches \cite{DBLP:journals/comsur/SonkolyCSNT21,DBLP:journals/csur/Ait-SalahtDL20}. In problem (P1), we do not explicitly model the constraint of edge servers' computation capacity, e.g., maximum CPU frequency. The main reason is, as discussed above, that we focus on caching service image files on edge servers, which prepares the computing environments in MEC \cite{DBLP:conf/nsdi/ChenLGL22}. More fine-grained optimization studies in computation offloading and container life-cycle management are also compatible with our work, which is an interesting direction for future research.

\subsection{Hardness of the TRSP Problem}
We first prove the TRSP problem is NP-hard by reduction from the 0-1 knapsack problem. The 0-1 knapsack problem considers that there are \(s\) items and each item \(i\) is equipped with weight \(w_i\) and value \(v_i\), and the objective is to maximize the total value of selected items under the constraint of total weight not exceeding \(R\). Next, we can construct an instance of TRSP problem by setting \(T=1\), \(n=1\), \(\lambda_{i,1}^{(1)}=(v_i+\alpha_i)/d_i\), \(r_i=w_i\) and \(R_{1}=R\), it can be easily shown that such TRSP problem is equal to the 0-1 knapsack problem. Therefore, TRSP is NP-hard.

The service refreshing cost in problem (P1) couples with the service placement and request offloading decisions both temporally (across multiple time slots) and spatially (across multiple edge servers). Specifically, there are \(Tns(n+2)\) decision variables in the problem, and \(2^{Tns(n+2)}\) possible combinations of these variables' value. Changing any one variable will impact on all other variables and the final transmission cost. Thus, the exhaustive search method's time complexity is \(O(2^{Tns(n+2)})\).

Next, if we remove the service placing and refreshing cost (e.g., \(\alpha_{i}\) and \(\beta_i\) are negligible compared with \(\lambda _{i,j}^{(t)}d_{i}\)), the service placement decisions at different time slots will become independent (i.e., changing one variable only impact other variables in the same slot). Specifically, the optimization problem of slot \(t\) is:
\begin{equation}
\label{subproblem_t}
\begin{split}
\min_{x,y} & \sum_{j = 1}^{n} \sum_{i=1}^{s} \sum_{k=0}^{n} \left ( 1-x_{i,j}^{(t)} \right ) \lambda _{i,j}^{(t)}d_{i}\cdot y_{i,j,k}^{(t)}\gamma _{j,k}  \\
\mathrm{s.t.}  & \ \left ( 4 \right ) \left ( 5 \right ) \left ( 6 \right ) \left ( 7 \right ).
\end{split}
\end{equation}

The problem (\ref{subproblem_t}) optimizes the service placement and request offloading decisions simultaneously. Note that this particular case of the TRSP problem is still NP-hard and has been studied in  \cite{DBLP:conf/icdcs/0001KWPS18,DBLP:conf/infocom/WangJHLM18}.

Therefore, the main challenge in solving the problem (P1) is its spatial-temporal coupling effect, i.e., the non-Markov property. When making decisions, the interaction between current decision variables and all other variables in different slots and edge servers should be considered. Among the existing methods that can handle coupled problems, the submodularity-based methods, e.g., greedy, local search, are widely utilized in service placement studies \cite{DBLP:conf/icdcs/0001KWPS18, DBLP:conf/infocom/0002ZHL021,DBLP:journals/tmc/MoubayedSHLB21}. However, the general TRSP problem is not submodular and cannot be solved by such methods, as discussed as follows.

\subsection{Special Case: Submodular Condition}
\begin{definition}
Given a finite ground set \(G\), a function \(F: 2^{G}\rightarrow \mathbb{R}\) is submodular if for any sets \(A\subseteq B\subseteq G\) and every element \(e \notin B\), it holds that:
\begin{equation}
    F\left ( A\cup \left \{ e \right \} \right )-F\left ( A \right )\geqslant F\left ( B\cup \left \{ e \right \} \right )-F\left ( B \right ).
\end{equation}
\end{definition}

Submodularity reflects the marginal diminishing effect, which is commonly utilized to handle complex combinatorial optimization problems. However, the submodularity of the TRSP problem only exists under strict conditions, i.e., the general TRSP problem is not submodular, and thus the widely-used submodularity-based methods do not suit in our problem.

\newtheorem{lemma}{Lemma}
\begin{lemma}
\label{submodular}
The TRSP problem is submodular if and only if:
\begin{itemize}
\item the traffic coefficient among edge servers and the cloud are equal: \(\gamma _{j,k}=1 ,\forall j\in \mathbb{N} ,\forall k\in \left \{ 0,\cdots ,n \right \}\) and
\item all services should refresh their database in every slot: \(LF_{i}=1, \forall i \in \left \{ 1,\cdots ,s \right \}\) and
\item for each service, the data volume of refreshing is equal to the service placement: \(\alpha_{i}=\beta_{i}, \forall i \in \left \{ 1,\cdots ,s \right \}\).
\end{itemize}
\end{lemma}
\begin{proof}
In the TRSP problem, the service placement strategies in all slots are represented by set \(G=\left \{e_{1,1}^{1},\cdots ,e_{s,n}^{T} \right \}\). If service \(i\) is placed at edge server \(j\) in slot \(t\), then element \(e_{i,j}^{t}\in G\), otherwise \(e_{i,j}^{t}\notin G\). We use the function \(F\) to represent the total transmission cost for the given service placement \(G\). Suppose the three conditions hold. If we add an element \(e_{i,j}^{t} \notin B\) into placement set \(A\) and \(B\) (\(A\subseteq B\)), the cost related to elements before and after slot \(t\) (\(e_{i,j}^{\tau}, \forall \tau \ne t\)) will not be affected, since all services will be refreshed in every slot. Other elements in slot \(t\) (\(e_{i',j'}^{t}, \forall j' \neq j \ or \ i' \neq i\)) are also not impacted. Besides, the traffic coefficient among edge servers and the cloud are the same, thus there is no difference between offloading requests to near edge servers and to the cloud. Finally, adding element \(e_{i,j}^{t}\) brings additional cost of \(\beta_i\) and reduces the cost of offloading user requests for service \(i\):
\begin{equation}
\begin{split}
F\left ( A\cup \left \{ e_{i,j}^{t} \right \} \right )-F\left ( A \right )=\lambda _{i,j}^{(t)}d_{i}-\beta _{i}, \\
F\left ( B\cup \left \{e_{i,j}^{t} \right \} \right )-F\left ( B \right )=\lambda _{i,j}^{(t)}d_{i}-\beta _{i},
\end{split}
\end{equation}
thus, \(F\) is submodular. If any condition does not hold, we can construct counterexamples easily, which are omitted due to the page limit.
\end{proof}

Although there are many other algorithms proposed by previous studies to handle the service placement problem, these methods do not suit well in our case due to the specific form of problem formulation. Based on the system model, we obtain an integer non-linear programming (INLP) problem with non-consistent objective functions (i.e., the indicator function that represents whether a service should be refreshed in Eq. (2)). The commonly used non-linear programming methods (e.g., convex optimization, linear programming transformation, and mathematical programming solver) cannot handle non-consistent functions, and thus can hardly work \cite{DBLP:journals/comsur/SonkolyCSNT21}. On the other hand, the algorithms to solve general mathematical programming problems (e.g., machine learning, evolutionary algorithm, and heuristic algorithm) cannot guarantee the performance theoretically \cite{DBLP:journals/comsur/SonkolyCSNT21}. Therefore, instead of using the above approaches, we devise a new method to handle the formulated problem with a theoretical performance guarantee.

\section{Algorithms and Solutions}
\label{sec_alg_solutions}
As discussed above, TRSP is a non-Markov problem and can hardly be solved by existing methods. To handle this problem, we first transform the original TRSP problem into a Markov problem by carefully designing the system state space, through which the decisions in different slots can be independent. Then we propose a dynamic programming-based centralized algorithm based on the transformation to obtain the optimal solution to the TRSP problem. Specifically, the centralized method's time complexity scales linearly (exponentially) with the number of time slots (edge servers and available services).

A distributed low-complexity discounted value approximation method (DVA) is also devised. In the temporal domain, DVA enables each edge server to predict current decisions' long-term transmission cost based on discounting the future slots' cost by introducing a temporal approximation factor. In the spatial domain, each edge server estimates the requests offloading cost based on their neighbors' state in the previous slot, where a spatial approximation factor is utilized to reflect the confidence of offloading requests to neighbors. Thus, the decisions of each edge server are decoupled in both temporal and spatial domains to realize a low-complexity and efficient algorithm. Furthermore, the approximation ratio of DVA is also analyzed theoretically.

\subsection{Centralized Algorithm}
Although the TRSP problem is NP-hard, its optimal solution can be obtained by transforming the original problem into a Markov shortest path problem. For notation simplicity, denote by \(Q^{(t)}=\{ X^{(t)},Y^{(t)},L^{(t)} \} \) the system state in slot $t$, where \(X^{(t)}= \{ x_{1,1}^{(t)},\cdots, x_{s,1}^{(t)} ;\cdots; x_{1,n}^{(t)},\cdots , x_{s,n}^{(t)}\} \) is the service placement strategy, \(Y^{(t)}= \{ y_{1,1,0}^{(t)},\cdots ,y_{s,1,n}^{(t)}; \cdots ;y_{1,n,0}^{(t)}, \\\cdots ,y_{s,n,n}^{(t)} \} \) denotes the request offloading decisions and \(L^{(t)}= \{ l_{1,1}^{(t)},\cdots ,l_{s,1}^{(t)}; \cdots; l_{1,n}^{(t)},\cdots ,l_{s,n}^{(t)} \} \) represents the remaining lifetime information. Then we can define the transmission cost of changing from state \(Q^{(t-1)}\) in slot \(t-1\) to state \(Q^{(t)}\) in slot \(t\) as:
\begin{equation}
\label{state_trans_cost}
D\left ( Q^{(t-1)},Q^{(t)} \right ) =\sum_{j=1}^{n} D_j^{(t)},
\end{equation}
which is hereinafter referred to as \textit{state transition cost}. For convenience, we initialize the state \(Q^0=\left \{ \mathbf{0} ,\mathbf{0} ,\mathbf{0}  \right \} \). Then, the objective of problem (P1) can be rewritten as:
\begin{equation}
D_{total}=\sum_{t=1}^{T} D\left ( Q^{(t-1)},Q^{(t)} \right ) .
\end{equation}

Note that the state transition cost is memoryless, i.e., \(D ( Q^{(t-1)},Q^{(t)} )\) is independent with \(D ( Q^{(\tau -1)},Q^{(\tau)} ),\) \(\forall \tau \in {1,\cdots, t-1}\). Therefore, given the information of the two states in adjacent time slots, the state transition cost can be calculated independently. Given the state \(Q^{(t)}\) in slot \(t\), we denote by \(A^*\left ( Q^{(t)} \right ) \) as the minimal transmission cost from slot \(t\) to \(T\) (i.e., the cost-to-go function in dynamic optimization studies), given as:
\begin{equation}
A^*\left ( Q^{(t)} \right )=\min_{Q^{(t+1)},\cdots,Q^{(T)}}\sum_{\tau=t}^T D\left ( Q^{(\tau)},Q^{(\tau+1)}\right ),
\end{equation}
which is not impacted by previous slots (states from slot \(1\) to \( t-1\)).

\begin{figure}[tbp]
\centering
\includegraphics[width=\linewidth]{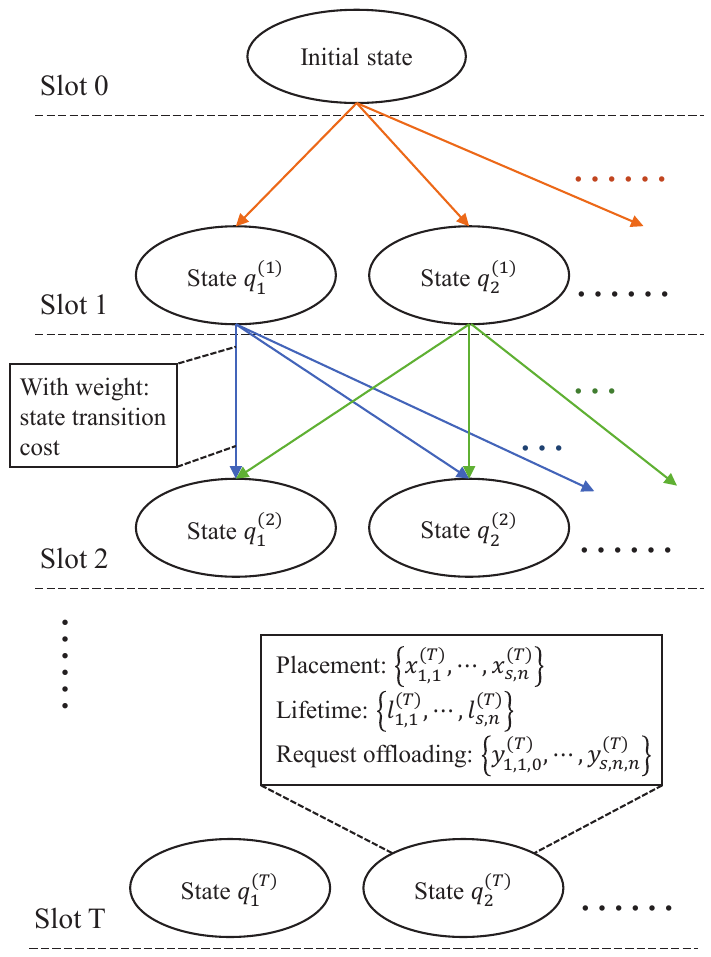}
\caption{Transforming the TRSP problem into a shortest path problem. Each node consists of service placement, request offloading decisions, and remaining lifetime information. The weight of paths represents the corresponding state transition cost.}
\label{shortest path}
\end{figure}

\begin{lemma}
The TRSP problem can be transformed into a shortest path problem.
\end{lemma}

\begin{proof}
We use nodes to represent the states in different slots and take the weight of a path to represent the corresponding state transition cost. Given all feasible nodes and paths, we can construct a graph, as illustrated in Fig. \ref{shortest path}.

As discussed above, every state in slot \(t\) is characterised by \(A^*\left ( Q^{(t)} \right ) \), and the optimal state \(q\) in slot \(t\) with minimal \(A^*\left ( q \right ) \) is not affected by previous slots. If \(A^*\left ( q \right ) \) corresponds to choosing state \(q_1\) in slot \(t+1\), \(q_2\) in slot \(t+2\), and so on, such state series will form an optimal path \(P_1= \{ q,q_1,q_2,\dots ,q_{T-t} \} \). Then, let path \(P_2=  \{ q_1,q_2,\dots ,\\ q_{T-t}  \} \) be part of \(P_1\), we can observe that \(P_2\) is also the optimal path starting from state \(q_1\) (otherwise \(P_1\) will choose a better sub-path). According to the Bellman equation, \(A^*\left ( Q^{(t)} \right ) \) can be formulated as:
\begin{equation}
\label{bellman}
\begin{split}
A^*\left ( Q^{(t)} \right ) =\min_{Q^{(t+1)}}\left [ D\left ( Q^{(t)},Q^{(t+1)} \right ) + A^*\left ( Q^{(t+1)} \right )\right ].
\end{split}
\end{equation}

Finally, the optimal solution to the TRSP problem is the shortest path starting from slot \(0\) (the initial state) to \(T\), which can be obtained via dynamic programming, as described in Algorithm \ref{optimal_alt}.
\end{proof}

Although the optimal solution can be obtained, the computation complexity is high due to the numerous states in each slot, i.e., the centralized method suffers from the curse of dimensionality \cite{bryson1998dynamic}. Let \(M=1+\max\left ( LF_{1},LF_{2},\cdots ,LF_{s} \right )\). The choices of value of service \(i\)'s remaining lifetime \(l_{i,j}^{(t)}\in \left \{ 0,1,\cdots ,LF_{i} \right \}\) cannot exceed \(M\). Thus, the amount of all possible sets of service's remaining lifetime at all edge servers cannot exceed \(M^{ns}\). Because a service is either placed or not, the amount of all possible sets of service placement decisions cannot exceed \(2^{ns}\). Similarly, the amount of all possible sets of request offloading decisions cannot exceed \(2^{n^2s}\). Thus, the amount of system states in slot \(t\) cannot exceed \(2^{ns\left ( 1+n \right ) }M^{ns}\). The computation complexity of finding such shortest path via dynamic programming is \(O\left ( 2^{2ns\left ( 1+n \right ) }M^{2ns}T \right )\), which is usually computationally intractable in practice.

\begin{algorithm}[t]
\caption{Centralized Optimal Algorithm}
\label{optimal_alt}
\begin{algorithmic}[1]
\State Construct all feasible states \(\left \{ q_1^t,q_2^t,\cdots  \right \} \) in each slot \(t\in \mathbb{T} \)
\State Construct feasible paths between states in adjacent slots
\State Compute state transition cost, allocate weight to each path
\For{each slot \(t\) from \(T\) to \(0\)}
\For{each state \(q\in\left \{ q_1^t,q_2^t,\cdots  \right \} \)}
\State Define states set \(P_q=\varnothing \)
\If{current slot is the final slot \(T\)}
\State Set \(A^*(q)=0\)
\State Set \(P_q=\left \{ q \right \} \)
\Else{}
\State Set $A^*(q) =$ \par \hspace{2.5em} $\min_{q'\in\left \{ q_1^{t+1},q_2^{t+1},\cdots  \right \} } \left [ D\left ( q,q' \right ) +A^*\left ( q' \right )  \right ]$ 

\State Get state \(q_{opt}^{t+1}=\) \par \hspace{2.5em} \(\arg \min_{q'\in\left \{ q_1^{t+1},q_2^{t+1},\cdots  \right \} } \left [ D\left ( q,q' \right ) +A^*\left ( q' \right )  \right ] \)
\State Set \(P_q=\left \{ q \right \} \cup  P_{q_{opt}^{t+1}}  \)
\EndIf
\EndFor
\EndFor
\State Get the optimal service placement and request offloading decisions according to set \(P_{Q^{0}}\)
\end{algorithmic}
\end{algorithm}

\subsection{Distributed Algorithm}
\label{sec_sta_val}
Due to the high time complexity, the centralized method cannot handle large-scale practical problems. In following discussions, we further design a distributed low-complexity discounted value approximation algorithm (DVA). We first decouple  (P1) in the temporal domain by estimating the future transmission cost, whereby (P1) is transformed into a series of problems (P2). Then, we decouple (P2) in the spatial domain by estimating the request offloading cost among edge servers, whereby each problem (P2) is transformed into multiple problems (P3). Finally, (P3) is addressed on each edge server in each slot to obtain the service placement and request offloading decisions. The approximation ratio of DVA is also analyzed theoretically.

\subsubsection{Temporal and Spatial Value Approximation}
A critical step in the centralized method is to compute \(A^*(Q^{(t)})\), i.e., the minimum future transmission cost given state \(Q^{(t)}\), which is extremely hard due to the enormous amount of states. Therefore, we apply value approximation to estimate the future transmission cost \cite{bertsekas1995dynamic}. Specifically, we compute the upper bound of \(A^*(Q^{(t)})\) to approximate it. Given system state \(Q^{(t)}=\left \{ X^{(t)} ,Y^{(t)}, L^{(t)}\right \} \) in slot \(t\), there is always a feasible strategy in future slots, i.e., keeping the service placement \(X^{(t)}\) and request offloading decisions \(Y^{(t)}\) unchanged. We denote by \(\hat{A} \left ( Q^{(t)} \right ) \) the transmission cost of such strategy after slot \(t\):
\begin{equation}
\label{appr_compu}
\begin{split}
\hat{A} \left ( Q^{(t)} \right ) =& \sum_{i=1}^{s} \sum_{j=1}^{n} \left [ x_{i,j}^{(t)}\beta _i\sum_{a=t}^{T} \varphi \left ( l_{i,j}^{(a)} \right )  \right ]
+ \\& \sum_{i=1}^{s} \sum_{j=1}^{n} \sum_{a=t}^{T} \left [ \left ( 1- x_{i,j}^{(t)}\right ) \lambda_{i,j}^{(a)} d_i\cdot \sum_{k=0}^{n} \gamma _{j,k}y_{i,j,k}^{(t)} \right ].
\end{split}
\end{equation}

Note that we add a factor \(x_{i,j}^{(t)}\) into the refreshing cost, which will not affect the computation (only these placed services have such refreshing cost). The remaining lifetime information of all services in all edge servers can also be obtained directly. According to the definition of \(A^*\left(Q^{(t)}\right)\), we have:
\begin{equation}
\hat{A} \left ( Q^{(t)} \right ) \ge A^*\left(Q^{(t)}\right).
\end{equation}
Thus, \(\hat{A} \left ( Q^{(t)} \right )\) is the upper bound of \(A^*\left(Q^{(t)}\right)\). Consequently, we can utilize it for approximation, and choose the state whose upper bound is lower.

From another perspective, the transmission cost of all future slots are discounted to current slot. Note that the slots in the near future are usually more important in practice, especially considering the accumulative error between the approximation value and actual value of \(A^*\left ( Q^{(t)} \right ) \). Hence, we introduce a temporal approximation factor \(\theta \in [0,1]\) to estimate the future transmission cost:

\begin{equation}
\label{discount_approx}
\begin{split}
\hat{A} \left ( Q^{(t)} \right ) =& \sum_{i=1}^{s} \sum_{j=1}^{n} \left [ x_{i,j}^{(t)}\beta _i\sum_{a=t}^{T} \theta ^{a-t} \varphi \left ( l_{i,j}^{(a)} \right )  \right ]
+ \\& \sum_{i=1}^{s} \sum_{j=1}^{n} \sum_{a=t}^{T} \left [ \theta ^{a-t} \left ( 1- x_{i,j}^{(t)}\right ) \lambda_{i,j}^{(a)} d_i\cdot \sum_{k=0}^{n} \gamma _{j,k}y_{i,j,k}^{(t)} \right ].
\end{split}
\end{equation}

If \(\theta \) is close to 1, more time slots will be taken into consideration; otherwise, the decisions will be made based on slots in the near future. Then, given state \(Q^{(t-1)}\) in slot \(t-1\), we can obtain the service placement and request offloading decisions in slot \(t\) by \(\min_{Q^{(t)}}\left [ D\left ( Q^{(t-1)},Q^{(t)} \right ) + \hat{A}\left ( Q^{(t)} \right )\right ]\), which is given as:
\begin{equation}
\label{temporal}
\begin{split}
(\mathrm{P2} )\ \min & \sum_{i = 1}^{s}\sum_{j = 1}^{n} \left ( 1-x_{i,j}^{(t-1)} \right ) x_{i,j}^{(t)}\alpha _i
+\\& \sum_{i = 1}^{s} \sum_{j = 1}^{n}\left [ x_{i,j}^{(t)}\beta _i\sum_{a = t}^{T}\theta ^{a-t} \varphi \left ( l_{i,j}^{(a)} \right )  \right ]
+ \\& \sum_{i = 1}^{s}\sum_{j = 1}^{n} \sum_{a = t}^{T} \left [ \theta ^{a-t}\left ( 1- x_{i,j}^{(t)}\right ) \lambda_{i,j}^{(a)} d_i\cdot \sum_{k=0}^{n} \gamma _{j,k}y_{i,j,k}^{(t)}  \right ]
\\ \mathrm{s.t.}  & \ (4)(5)(6)(7).
\end{split}
\end{equation}

Problem (P2) is still NP-hard in the general case, where the service placement and request offloading variables are coupled. To further reduce the complexity, problem (P2) is decoupled in spatial domain by using approximation of the request offloading cost (i.e., the spatial value approximation).

In the problem (P2), \(\sum_{k = 0}^{n} \gamma _{j,k}y_{i,j,k}^{(t)}\) represents the traffic coefficient of offloading service \(i\)'s requests from edge server \(j\). Intuitively, if edge server \(j\)'s neighbors have placed service \(i\) in slot \(t-1\), it is likely that the service \(i\) will still be placed in edge server \(j\)'s neighborhood in slot \(t\) (maybe not the same neighbor). Hence, we introduce a spatial approximation factor \(\delta _{i,j}\ge 1\) to estimate the request offloading cost among edge servers:
\begin{equation}
\gamma _{i,j}^{pre}=\delta _{i,j}\gamma _{j}^{i,min},
\end{equation}
where \(\gamma _{j}^{i,min}=\min_{x_{i,k}^{(t-1)}=1, j\ne k} \gamma _{j,k}\). A small \(\delta _{i,j}\) indicates edge server \(j\) is optimistic that its nearest neighbors (edge servers with smallest \(\gamma _{j,k}\)) will place service \(i\). Similarly, a small \(\gamma _{j}^{i,min}\) means \(j\)'s nearest neighbors have cached service \(i\) in last slot.

By introducing the spatial approximation factor, we remove the request offloading variables \(y_{i,j,k}^{(t)}\), thus the service placement decisions of different edge servers are decoupled. Therefore, we can further decompose problem (P2) into \(n\) independent sub-problems, namely, each edge server \(j\) only computes for its own:
\begin{equation}
\begin{split}
\mathrm{(P3)} \ \min_{X^{(t)}} & \sum_{i = 1}^{s} \left ( 1-x_{i,j}^{(t-1)} \right ) x_{i,j}^{(t)}\alpha _i
+\\& \sum_{i = 1}^{s} \left [ x_{i,j}^{(t)}\beta _i\sum_{a = t}^{T} \theta ^{a-t}\varphi \left ( l_{i,j}^{(a)} \right )  \right ]
+ \\& \sum_{i = 1}^{s} \sum_{a = t}^{T} \left [\theta ^{a-t} \left ( 1- x_{i,j}^{(t)}\right ) \lambda_{i,j}^{(a)} d_i\cdot \gamma _{i,j}^{pre} \right ]
\\ \mathrm{s.t.}  & \sum_{i = 1}^{s} x_{i,j}^{(t)}r_i\le R_j.
\end{split}
\end{equation}

Problem (P3) is a 0-1 knapsack problem, which can be solved by dynamic programming with time complexity of \(O(Rs)\), where \(R=\max_{j\in\left \{ 1,\cdots ,n \right \} }R_j\). 

Finally, we present a discounted value approximation algorithm, DVA, which is summarized in Algorithm \ref{DVA}. In DVA, every edge server only needs to solve problem (P3) to obtain the current slot's service placement decision by combining the temporal and spatial approximations. Then these edge servers can exchange their service placement information with their neighbors\footnote{the communication range can be obtained via clustering edge servers or splitting the edge network \cite{DBLP:journals/tmc/NdikumanaTHHSNH20, DBLP:journals/tmc/ChenSZX21}, which is not the focus of this work.} (for edge nodes far away, the request offloading costs may be as expensive as the cloud), and the request offloading strategy can be obtained accordingly. Thus, our method is suitable for large-scale edge networks naturally.

\begin{algorithm}[t]
\caption{Discounted Value Approximation Algorithm}
\label{DVA}
\begin{algorithmic}[1]
\State Initialize state \(Q^{(0)}=\left \{ \mathbf{0} ,\mathbf{0} ,\mathbf{0}  \right \} \)
\State In each slot \(t\in\mathbb{ T} \)
\State \hspace*{0.5cm} Every edge server \(j\in\mathbb{ N} \) solves the problem (P3) and obtains its service placement decisions \(\left \{ x_{1,j}^{(t)},\cdots ,x_{s,j}^{(t)} \right \} \);
\State \hspace*{0.5cm} Every edge server exchanges its service placement information with neighbors;
\State \hspace*{0.5cm} Every edge server \(j\in\mathbb{ N} \) obtains its request offloading decisions \(\left \{ y_{0,j,0}^{(t)},\cdots ,y_{s,j,n}^{(t)} \right \} \);
\State \hspace*{0.5cm} Every edge server \(j\in\mathbb{ N} \) updates its services remaining lifetime information \(\left \{ l_{1,j}^{(t)},\cdots ,l_{s,j}^{(t)} \right \} \).
\end{algorithmic}
\end{algorithm}

\subsubsection{Approximation Ratio Analysis}
\label{appro}
In this part, we analyze the approximation ratio of the DVA algorithm, i.e., the performance gap between DVA and the optimal method. For the convenience of analysis, we assume all services have the same size \(r_i=r,\forall i\in\left \{ 1,\cdots ,s \right \} \), and all services are worthy to be placed, i.e., the cost of offloading service \(i\)'s requests is higher than the placement cost \(\lambda _{i,j}^{(t)}d_i\gamma _{j}^{min}>\alpha _i,\) \(\forall i\in \mathbb{S} , \forall j \in \mathbb{N} , \forall t \in \mathbb{T} \), where \(\gamma _{j}^{min}=\min_{k\ne j } \gamma _{j,k}\). Thus, for edge server \(j\), it will always place \(\left \lfloor R_j/r \right \rfloor\) services. Besides, we set the spatial approximation factor \(\delta _{i,j}=\frac{1}{\gamma _{j}^{i,min}} \) and the temporal approximation factor \(\theta =1.0\). Such settings indicate that the predicted traffic coefficient for offloading requests will always be \(\gamma _{j,0}=1\) for all edge servers.

\begin{lemma}\label{DVA_cost}
The total transmission cost of DVA's solution, \(D_{A}\), satisfies:
\begin{equation}
\begin{split}
D_A\le \sum_{j=1}^{n} \left [ \sum_{i\in u_j}\left ( \alpha _i+\left \lfloor \frac{T}{LF_i}  \right \rfloor \beta _i \right ) 
+ \sum_{i\notin u_j} \sum_{t=1}^{T} \lambda _{i,j}^{(t)}d_i\right ],
\end{split}
\end{equation}
where \(u_j=\arg \max_{\left \lfloor R_j/r \right \rfloor } \sum_{t=1}^{T} \lambda _{i,j}^{(t)}d_i\) represents the top \(\left \lfloor R_j/r \right \rfloor\) services with largest volume of requests offloading data across all slots on edge server \(j\).
\end{lemma}
\begin{proof}
In the first slot, DVA will make service placement decisions \(f_{1}=\left \{ x_{1,1}^{(1)},\cdots ,x_{s,n}^{(1)}\right \}\) for all edge servers. We denote by \(D\left ( f_{1},f_{1},\cdots ,f_{1} \right )\) as the total transmission cost of keeping \(f_{1}\) unchanged (and edge servers only offload requests to the cloud) in all following slots. Next, in the second slot, DVA will make placement decisions \(f_{2}\). According to the mechanism of DVA, the cost \(D\left ( f_{1},f_{2},\cdots ,f_{2} \right )\) of keeping \(f_{2}\) unchanged after slot 2 can not exceed \(D\left ( f_{1},f_{1},\cdots ,f_{1} \right )\). Otherwise, \(f_{1}\) will not be replaced by \(f_{2}\). Similarly, in slot 3, \(D\left ( f_{1},f_{2},f_{3},\cdots ,f_{3} \right )\le D\left ( f_{1},f_{2},\cdots ,f_{2} \right )\). Finally, we will have \(D\left ( f_1,\cdots ,f_T \right ) \le D\left ( f_1,\cdots ,f_1 \right ) \). Note that all the above costs correspond to the case that user requests are all offloaded to the cloud. However, edge servers will offload requests to their neighbors in practice (if possible). Thus, the actual cost of DVA, \(D_A\), satisfies:
\begin{equation}
D_A\le D\left ( f_1,\cdots ,f_T \right ) \le D\left ( f_1,\cdots ,f_1 \right ) .
\end{equation}

Suppose the services in set \(u_j=\arg \max_{\left \lfloor R_j/r \right \rfloor } \sum_{t=1}^{T} \lambda _{i,j}^{(t)}d_i\) are placed on edge server \(j\) in all slots \(\left \{ 1,\cdots ,T \right \} \), and other requests are offloaded to the cloud. The total transmission cost of such strategy will be \(D_S\):
\begin{equation}
D_S=\sum_{j=1}^{n} \left [ \sum_{i\in u_j}\left ( \alpha _i+\left \lfloor \frac{T}{LF_i}  \right \rfloor \beta _i \right ) 
+ \sum_{i\notin u_j} \sum_{t=1}^{T} \lambda _{i,j}^{(t)}d_i\right ].
\end{equation}
Given that \(\theta =1\) and \(f_{1}\) is the solution of problem (P2), we have:
\begin{equation}
D\left ( f_{1},f_{1},\cdots ,f_{1} \right )\leqslant D_S.
\end{equation}
Thus, Lemma \ref{DVA_cost} is proved.
\end{proof}

\begin{lemma}\label{OPT_cost}
The total transmission cost of the optimal solution, \(D_O\), satisfies:
\begin{equation}
D_O\ge \sum_{j=1}^{N} \sum_{t=1}^{T} \sum_{i \notin v_j^{(t)}}\lambda _{i,j}^{(t)}d_i\gamma _{j}^{min} ,
\end{equation}
where \(v_j^{(t)}=\arg \min_{\left \lfloor R_j/r \right \rfloor }\lambda _{i,j}^{(t)}d_i\) represents the top \(\left \lfloor R_j/r \right \rfloor\) services with largest volume of requests offloading data in slot \(t\) on edge server \(j\).
\end{lemma}
\begin{proof}
We use \(D_{ignore}\) to denote the total transmission cost of placing services \(v_1^{(1)}\) in slot 1 on edge server 1, \(v_2^{(1)}\) in slot 1 on edge server 2, etc. Besides, we ignore the service placement and refreshing cost, and every edge server can offload requests to its nearest neighbor, therefore:
\begin{equation}
D_{ignore}=\sum_{j=1}^{N} \sum_{t=1}^{T} \sum_{i \notin v_j^{(t)}}\lambda _{i,j}^{(t)}d_i\gamma _{j}^{min} .
\end{equation}

According to Eq. (\ref{trans_cost_j}):
\begin{equation}
D_j^{(t)}\ge \sum_{i \notin v_j^{(t)}}\lambda _{i,j}^{(t)}d_i\gamma _{j}^{min}.
\end{equation}
Since there are additional service placement and refreshing costs in the optimal solution, we have:
\begin{equation}
D_O=\sum_{j=1}^{N} \sum_{t=1}^{T}D_j^{(t)}> D_{ignore}.
\end{equation}
\end{proof}
According to Lemmas \ref{DVA_cost} and \ref{OPT_cost}, the approximation ratio of DVA satisfies:
\begin{equation}
\label{ar}
\begin{split}
\frac{D_A}{D_O}&< \frac{D_S}{D_{ignore}} \\& \leqslant \frac{A+\frac{1}{T}\sum_{j=1}^{N} \sum_{t=1}^{T} \sum_{i \notin u_j}\lambda _{i,j}^{(t)}d_i}{\frac{1}{T}\sum_{j=1}^{N} \sum_{t=1}^{T} \sum_{i \notin v_j^{(t)}}\lambda _{i,j}^{(t)}d_i\gamma _{j}^{min}} ,
\end{split}
\end{equation}
where \(A=\sum_{j=1}^{n} \sum_{i\in u_j} \left ( \frac{\alpha _i}{T} +\frac{\beta _i}{LF_i}  \right ) \). The analysis result indicates that the approximation ratio of DVA is small when the popularity of services is stable across the whole time window, i.e., the service set \(u_j\) is similar with \(v_j^{(t)}, \forall t\).

\section{Performance Evaluation}
\label{sec_simulation}

In this section, we evaluate the performance DVA thoroughly. Firstly, the approximation ratio of DVA (gap between DVA and the optimal results) is evaluated in a small-scale system. Then, the performance of DVA is compared with the state-of-the-art methods in extensive real-trace simulations and 5G testbed experiments. In following experiments, the length of one time is set as one minute.

Unless otherwise stated, the image file size \(r_{i}\) of each service is drawn from \([1,3]\) GB randomly, and the data volume of service placement, \(\alpha_{i}\) equals \(r_{i}\). The data volume of offloading one request for service \(i\), \(d_{i}\), is randomly drawn from \([0.05, 0.1]\) times of \(r_{i}\). The data volume of refreshing, \(\beta_{i}\), is uniformly generated from \([0.5, 0.8]\) times of \(r_{i}\). The lifetime \(LF_{i}\) follows an exponential distribution with a mean value of 4 minutes. Every edge server covers 1000 users, where the user requests are generated randomly according to Zipf distribution with shape parameter 0.6. Besides, the probability of re-ranking the popularity of services is 0.3 in each slot, which reflects the practical dynamic services popularity. The predicted request information is not always accurate. Hence, the prediction error is introduced in our simulations and testbed experiments, i.e., the predicted request number is randomly drawn from \([0.7, 1.3]\) times of the actual value.

The DVA algorithm is compared with the following baselines:
\begin{enumerate}
\item \textbf{Popular}: In each slot, services receiving the largest number of requests (\(\lambda_{i,j}^{(t)}\)) are placed on edge server \(j\).
\item \textbf{Greedy}: In each slot, services with the largest ratio (\(\lambda_{i,j}^{(t)}d_{i}/r_{i}\)), i.e., ratio of offloading all requests' data volume to service image size, are placed on edge server \(j\).
\item \textbf{Relaxation-rounding (RR)}: We relax the binary constraints in the TRSP problem, assume every service brings \(\beta_i/LF_i\) refreshing data in each slot, and every edge server can offload requests to the nearest neighbor. We use the Gurobi optimizer to solve the relaxed problem, where the solution is further rounded into a feasible solution for the TRSP problem.
\item \textbf{Optimal}: The optimal solution obtained by the centralized algorithm.
\end{enumerate}

The Greedy, Popular and RR methods have been utilized in recent service placement studies. Specifically, the Greedy method can achieve (\(1-1/e\))-optimal when the optimization problem is submodular and monotone \cite{DBLP:conf/icdcs/0001KWPS18, DBLP:conf/infocom/0002ZHL021,DBLP:journals/tmc/MoubayedSHLB21}, while the RR method can also obtain near-optimal solutions with a high probability \cite{DBLP:conf/infocom/PoularakisLTTT19,mitzenmacher2017probability,DBLP:conf/mobihoc/PoularakisLTT20,DBLP:conf/infocom/XuZCLXZ20}. Hence, these algorithms are adopted as the state-of-the-art baselines \footnote{Other (meat-)heuristic algorithms or reinforcement learning methods are not chosen due to the lack of theoretical performance guarantee.}.

\subsection{Comparison with Optimal Results}
\label{com_opt}

\begin{figure}[t]
\centering
\includegraphics[width=\linewidth]{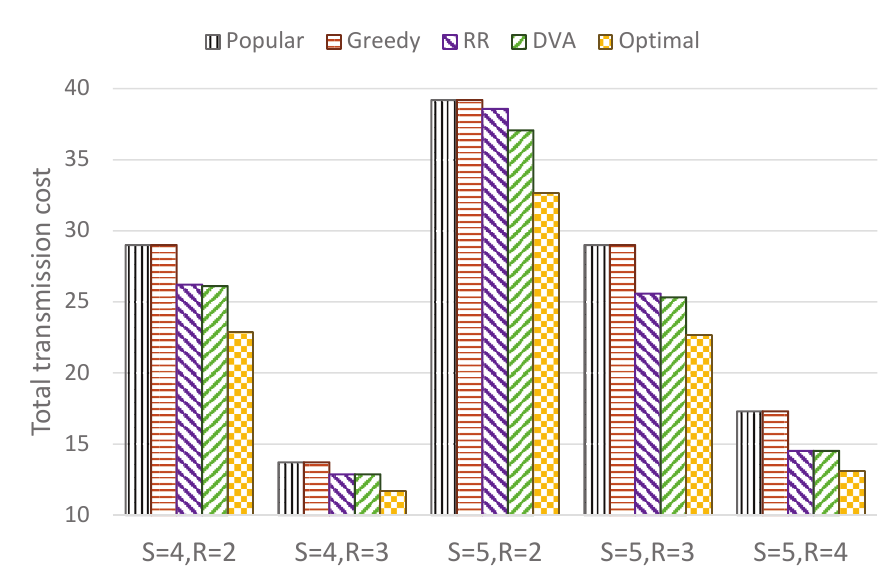}
\caption{Comparison with optimal results.}
\label{optimal}
\end{figure}

In this part, we compare DVA with the optimal results and evaluate its approximation ratio, i.e., the ratio of DVA's transmission cost to the optimal cost. However, the centralized algorithm's exponential complexity makes it extremely hard to obtain multiple edge servers' service placement decisions. For example, in a system with only two edge servers and two services, if \(LF_{1}=LF_{2}=1\), there will be \(2^{16}\) states in each slot, and the complexity of obtaining the optimal solution is \(O(2^{32}T)\). Therefore, we consider a system with only 1 edge server, 5 services, and 10 slots. In such scenario, it can be envisioned that edge servers only offload requests to the cloud, i.e., \(\gamma _{j,k}=1,\forall j\in \mathbb{N} ,\forall k\in \left \{ 0,\cdots ,n \right \} \). We set \(r_i =1\) GB and \(LF_i=2\) minutes for all services. We set the temporal approximation factor \(\theta=0.5\) and the spatial approximation factor \(\delta=1.0\) in DVA. Fig. \ref{optimal} lists the comparison results. The horizontal axis denotes the services number \(S\) and edge server storage \(R\). As can be observed, DVA always performs better than all the baselines and is close to the optimal results.

According to section \ref{appro}, the analytical ratio of DVA's cost to the optimal should be in [1.0, 1.39], and the actual ratio in the above simulations ranges from 1.10 to 1.14. As to the scenario with multiple edge servers, we can obtain the lower bound of the optimal solution according to Lemma \ref{OPT_cost}. The theoretical ratio of DVA to the lower bound should be in [1.0, 4.15], while the actual ratio is about 3.21. In conclusion, the theoretical analysis of DVA's approximation ratio is validated by the simulation results. In the above simulations, the performance of RR method is close to DVA. The main reason is that the relaxation strategy in RR suits well in small scale scenarios. However, in large-scale simulations, we find that RR performs much worse.

\begin{table}[]
\centering
\caption{Running time (in milliseconds) of different methods.}
\label{running time}
\begin{tabular}{@{}cccc@{}}
\toprule
Services number & DVA   & Greedy/Popular & RR      \\ \midrule
100             & 25.6  & 0.8            & 1561.9  \\
500             & 125.3 & 4.9            & 8169.1  \\
1000            & 249.5 & 9.4            & 16432.9 \\ \bottomrule
\end{tabular}
\end{table}

Next, the running time of different methods is analyzed in table \ref{running time}. All the simulations are implemented on a machine with i5-8400 CPU and 16 GB RAM. As can be seen, DVA's running time scales linearly with the services number. Specifically, DVA can handle 1000 services within 250 ms, i.e., 4‰ of a time slot, which indicates that DVA is suitable for large-scale scenarios. DVA and the Greedy/Popular method can generate the service placement decisions in a slot-by-slot manner, while the RR method needs to compute for all future slots simultaneously. Thus, the running time of RR is significantly higher than other methods.

\subsection{Performance Comparison and Analysis}
\label{pere}
In this part, we set 16 edge servers that form a \(4\times 4\) grid, and the traffic coefficient between edge servers is given by \(\gamma_{j,k}=0.3\times dis(j,k)\), where \(dis(j,k)\) represents the hops between \(j\) and \(k\). Hence, the cost of offloading requests to a 4-hop away node will be higher than that to the cloud. For simplicity, all edge servers have equal storage size, i.e., \(R_j=100\) GB, \(\forall j\in \left \{ 1,\cdots ,n \right \}\). We consider 100 services and 20 time slots in the following simulations.

\begin{figure}[t]
\centering
\includegraphics[width=\linewidth]{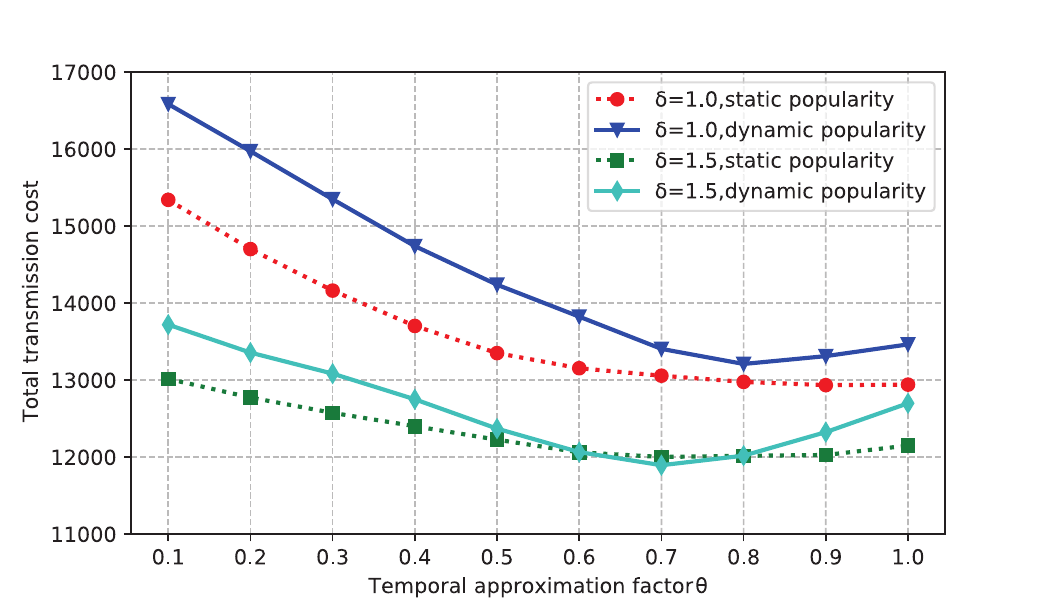}
\caption{Impact of the temporal approximation factor \(\theta\).}
\label{temporal factor}
\end{figure}

We first analyze the impact of the temporal approximation factor \(\theta\) and spatial approximation factor \(\delta_{i,j}\) in DVA. Recall that the transmission cost in the future slots is discounted to the current slot by using $\theta$. Besides, the cost of offloading requests among edge servers is estimated by $\delta$, which represents the unit cost of offloading one GB of data. By adjusting $\delta$ and $\theta$, we can predict the future transmission cost relatively accurately. Then, the optimal strategy in the current slot can be obtained by comparing the predicted transmission cost.

For simplicity, all edge servers share a fixed spatial approximation factor \(\delta_{i,j}=\delta,\forall i\in \{ 1,\cdots ,s \},\) \(\forall j\in \{ 1,\cdots ,n \}\). Fig. \ref{temporal factor} demonstrates the impact of different temporal approximation factors. Note that the \textit{static popularity} corresponds to the unchanged popularity rank of services across all slots, and the \textit{dynamic popularity} represents dynamic service popularity as mentioned above. Unless otherwise stated, we consider the dynamic popularity scenario in the following simulations since it can reflect the practical applications. As can be seen, with static popularity, large temporal approximation factors always perform better, which means that considering more slots can significantly reduce the cost. However, with dynamic popularity, the transmission cost will increase with too large or small temporal approximation factors.

\begin{figure}[t]
\centering
\includegraphics[width=\linewidth]{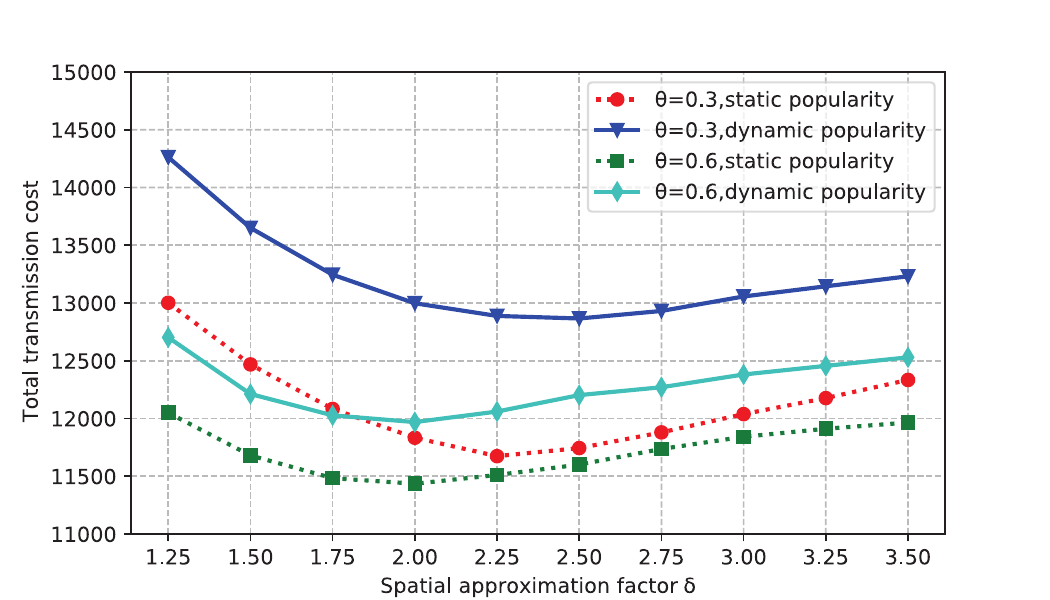}
\caption{Impact of the spatial approximation factor \(\delta\).}
\label{spatial factor}
\end{figure}

The impact of the spatial approximation factor \(\delta\) is illustrated in Fig. \ref{spatial factor}. The total transmission cost decreases at first and then increases when the factor is large enough. The main reason is that a small \(\delta\) represents the edge server is quite confident that its requests can be offloaded to the nearest neighbors, which may deviate significantly from reality. In contrast,  the possibility of offloading requests is much higher with a large \(\delta\), which can approximate the actual state better. However, when the factor is large enough, the predicted request offloading cost will also deviate from reality. Therefore, we set \(\theta=0.6\) and \(\delta=2.0\) in following simulations.

\begin{figure}[t]
\centering
\includegraphics[width=\linewidth]{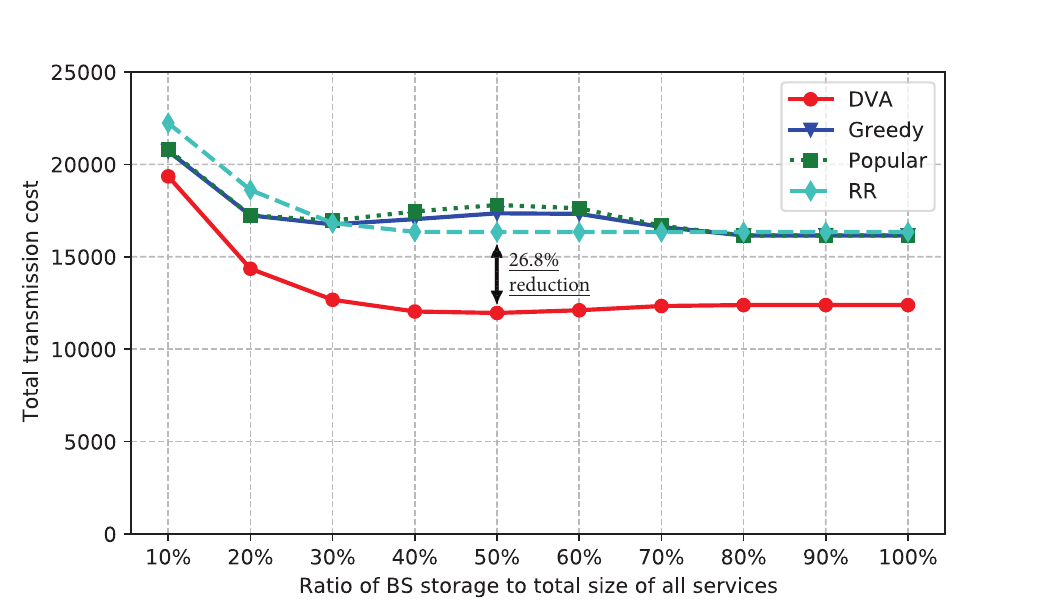}
\caption{Impact of edge servers' storage capacity.}
\label{storage 1}
\end{figure}

The impact of edge server storage is shown in Fig. \ref{storage 1}, where the horizontal axis represents the ratio of edge server storage to the total size of all service image files. The proposed DVA algorithm consistently outperforms the baseline methods and can save 26.8\% transmission cost when the ratio is 50\%. For DVA, the total cost decreases with the  storage capacity and then reaches a relatively stable value. The main reason is that more services can be placed at the network edge with larger storage. Thus, more user requests can be handled by edge servers, which reduces the total cost. However, when the storage is large enough, DVA tends to offload some unpopular services' requests other than place these services. The RR method can also achieve a stable value. Nevertheless, RR focuses on a simplified problem, and its solution fits not well in the original TRSP problem and leads to an increased cost.

\begin{figure}[t]
\centering
\includegraphics[width=\linewidth]{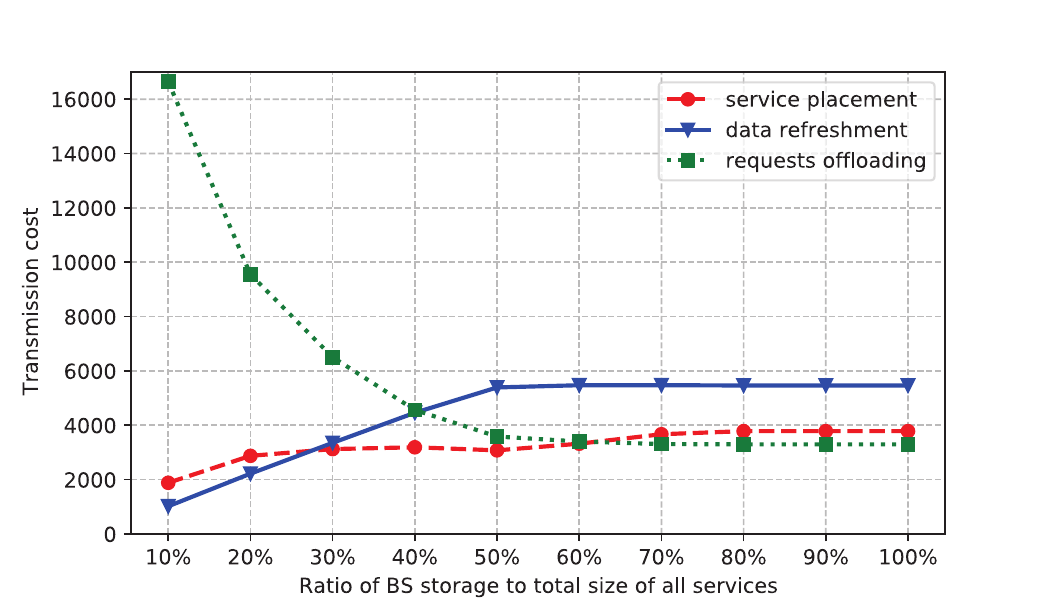}
\caption{Components in the total transmission cost of DVA.}
\label{DVA components}
\end{figure}

\begin{figure}[t]
\centering
\includegraphics[width=\linewidth]{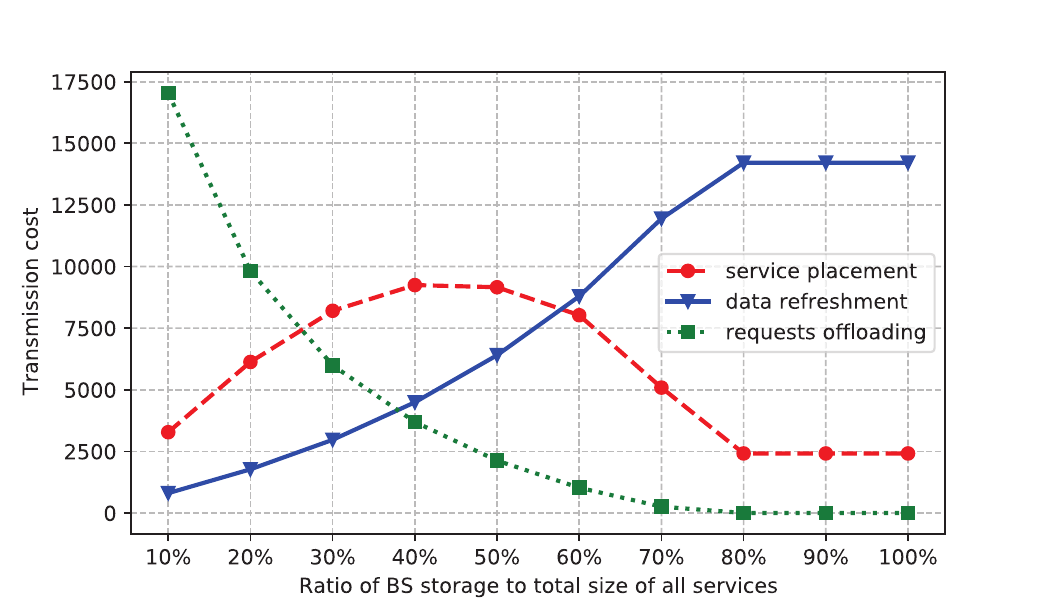}
\caption{Components in the total transmission cost of the greedy method.}
\label{greedy components}
\end{figure}

The main reason for the greedy (popular) method's bad performance is that it tends to place all services. Due to the fluctuating service popularity, the greedy (popular) method frequently places new services; thus, the gain of serving requests at the edge is eliminated by additional service placement cost, which leads to an increased total cost when the ratio is between 30\% and 70\%. When the storage is large enough, most services will be placed at edge servers; thus, the service placement cost is replaced by refreshing cost, which leads to a smaller total transmission cost. We list the components of the total transmission cost of DVA and the greedy method in Fig. \ref{DVA components} and Fig. \ref{greedy components}. These two figures can verify the above analysis.

\begin{figure}[t]
\centering
\includegraphics[width=\linewidth]{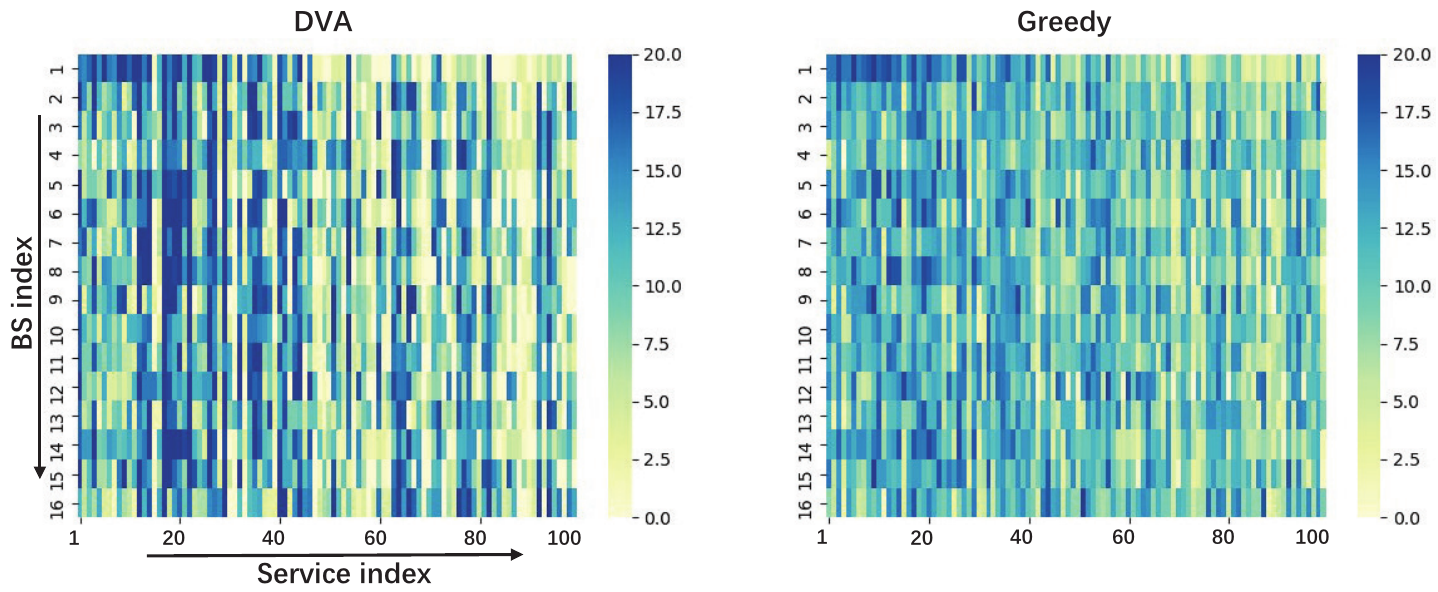}
\caption{Service placement decisions of DVA and the greedy method across 20 slots, point with light color indicates the service (column) is rarely placed on the edge server (row).}
\label{heatmap}
\end{figure}

Fig. \ref{heatmap} demonstrates the service placement decisions of DVA and the greedy method across all 20 slots. If a service \(i\) is placed in all slots on edge server \(j\), the (\(j, i\))-th point in the figure will be characterized with the deepest color. As can be observed, the color distribution of the greedy method is much more even than DVA. Such phenomenon indicates the greedy method will frequently place new popular services, bringing higher transmission cost. In contrast, DVA prefers to keep the service placement decisions unchanged in the long term, which achieve the balance among service placement, refreshing and offloading cost.

\begin{figure}[t]
\centering
\includegraphics[width=\linewidth]{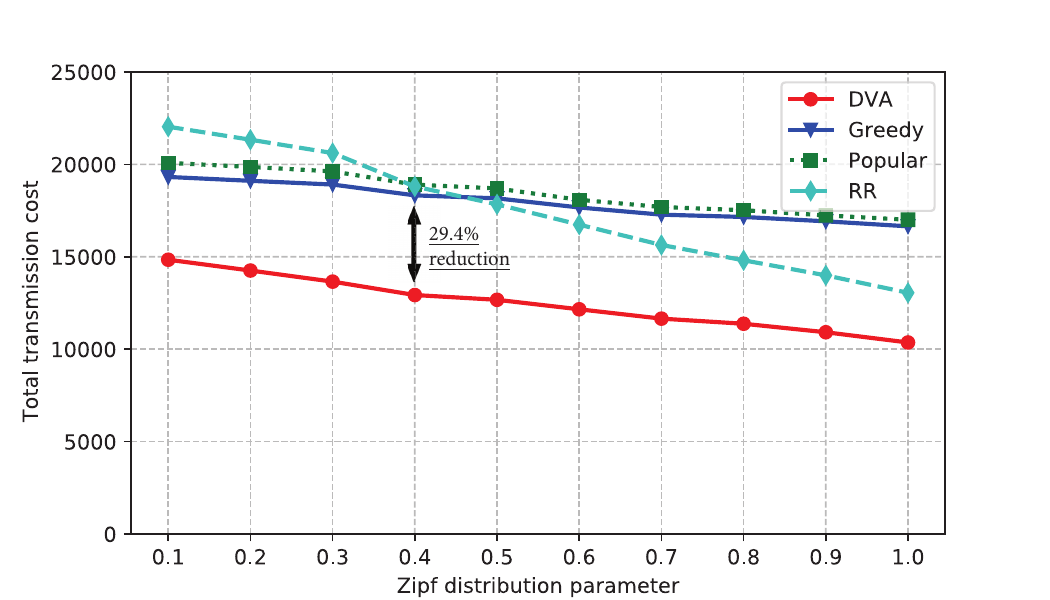}
\caption{Impact of Zipf distribution parameter.}
\label{zipf 1}
\end{figure}

The influence of user request concentration is also studied by varying the Zipf distribution parameter, as shown in Fig. \ref{zipf 1}. The total transmission cost is shown to decrease with the increasing parameters in all methods. In this case, user requests are more concentrated, and the most popular services will be placed to achieve a high hit rate. The DVA algorithm also outperforms the baselines with up to 29.4\% transmission cost reduction.

\subsection{Real-Trace Simulations}
\label{rts}
In this part, we consider 4 edge servers located alongside the Shanghai Yan An Elevated Road. The easternmost edge server is located at Jing'an Temple Square (31.22N, 121.45E), and the distance between adjacent edge servers is 2000 m. The coverage radius of each edge server is 1000 m. All edge servers have equal storage size, \(R_j=80\) GB, \(\forall j\in \left \{ 1,\cdots ,4 \right \}\). The real taxi trace data is used for simulation \cite{DBLP:conf/kdd/LiuLNFL10}. The trace data is updated every minute and consists of taxi ID, longitude, latitude, velocity, driving direction. The number of taxis in coverage is around 52 taxis per minute per edge server on average. Each taxi generates 12 requests per minute. The system consists of 100 services and 60 slots. We set \(\theta=0.6\) and \(\delta=2.0\) in following simulations. Other settings are kept the same with section \ref{pere}.

\begin{figure}[t]
\centering
\includegraphics[width=\linewidth]{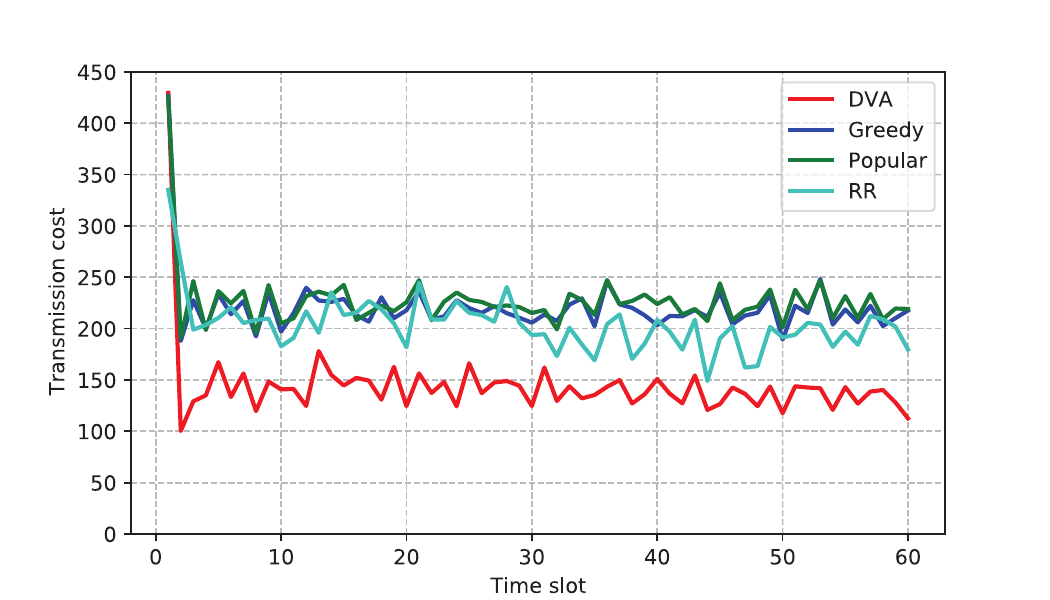}
\caption{Transmission cost of different time slots in real-trace simulations.}
\label{slot cost}
\end{figure}

Firstly, we analyze the transmission cost in different slots of all methods in Fig. \ref{slot cost}. As can be observed, DVA always performs better than other methods across all time slots. Note that all methods have a high cost in the first slot because all new services should be placed in the first slot.

\begin{figure}[t]
\centering
\includegraphics[width=\linewidth, trim=0 0 0 30, clip]{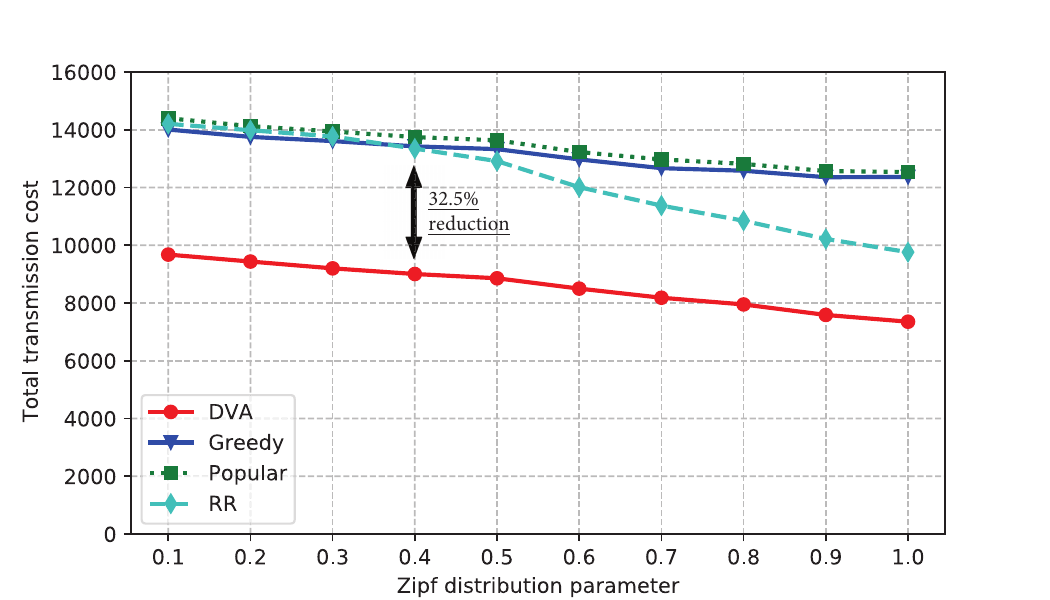}
\caption{Impact of Zipf distribution parameter in real-trace simulations.}
\label{zipf 2}
\end{figure}

Next, we explore the impact of the Zipf distribution parameter. Fig. \ref{zipf 2} shows the results, which is similar to Fig. \ref{zipf 1}. DVA consistently outperforms the other methods with up to 32.5\% transmission cost reduction, which shows the good performance of the DVA algorithm.

\subsection{5G-testbed Experiments}

\begin{figure}[t]
\centering
\includegraphics[width=\linewidth]{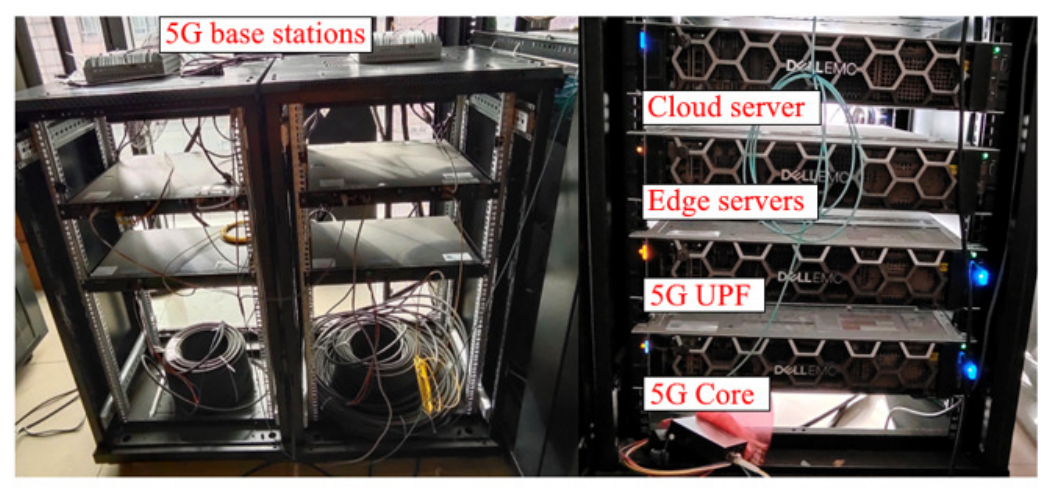}
\caption{5G service placement testbed.}
\label{5gt}
\end{figure}

In this part, we evaluate DVA's performance using a 5G testbed. As shown in Fig. \ref{5gt}, the testbed consists of two 5G white-box base stations, two edge servers and one cloud server. The edge/cloud servers are built on Dell PowerEdge R740. Service images and containers in edge/cloud servers are managed using Docker \footnote{https://www.docker.com/}. Placed services are registered using ZooKeeper \footnote{https://zookeeper.apache.org/}. User requests received by base stations are forwarded to edge/cloud servers by the 5G User Plane Function (UPF).

\begin{figure}[t]
\centering
\includegraphics[width=0.9\linewidth]{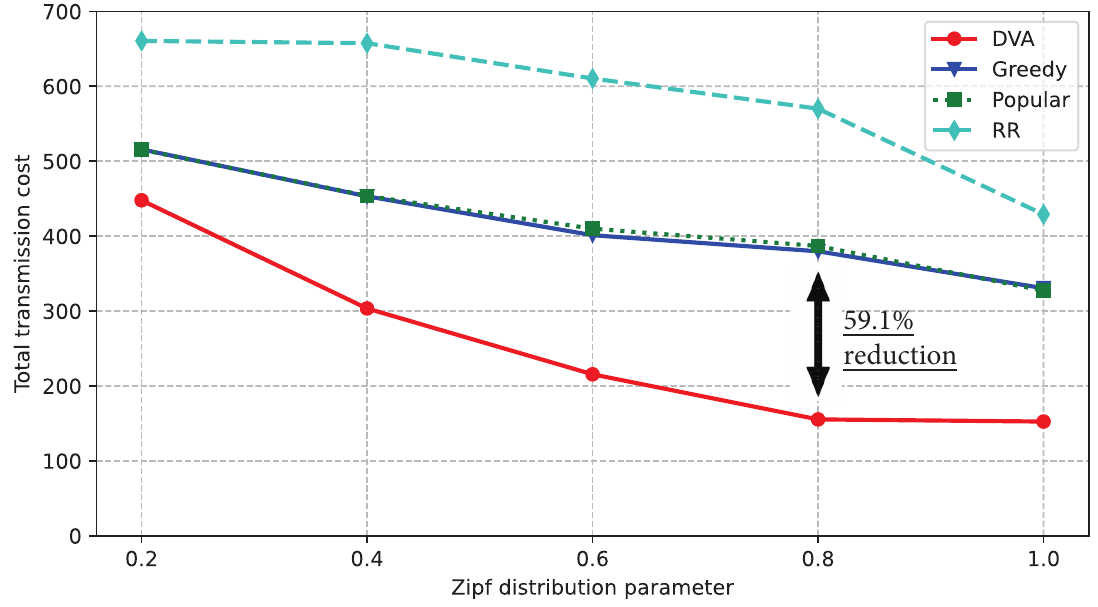}
\caption{Impact of Zipf distribution parameter in 5G testbed experiments.}
\label{tzipf}
\end{figure}

Due to the throughput limitations of the network interface cards in our testbed, we set the mean service image file size as 0.3 GB. All edge servers have equal storage size, \(R_j=15\) GB. We consider a system of 100 services and 20 slots, the lifetime \(LF_{i}\) is 1 minute for all services. Other settings are kept the same with section \ref{pere} and \ref{rts}.

The impact of the Zipf distribution parameter is shown in Fig. \ref{tzipf}. As can be seen, DVA outperforms all baseline methods with up to 59.1\% transmission cost reduction, which shows the usability and superiority of DVA in practical edge computing systems. On the other hand, the RR method is relatively unstable due to its randomized decision making strategy, which performs much worse in practical systems.

In summary, compared with the baseline methods, DVA can effectively balance the service placement, refreshing, and offloading costs, which makes it a superior method for placing timely refreshing applications in practical systems.

\section{Conclusions and Future Work}
\label{sec_conclusions}
This paper investigates the problem of placing timely refreshing services at the network edge. Aiming at minimizing the backhaul transmission cost, we formulate an integer non-linear programming problem and prove its hardness. The main difficulty in the formulated problem is the complex spatial-and-temporal coupling property brought by the practical maintaining cost of services. To solve the problem, we first decouple it in the temporal domain by transforming it into a shortest-path problem. Then, a dynamic programming-based algorithm is proposed to obtain the optimal solution with exponential time complexity. To reduce the computational complexity, we further design a light-weighted algorithm, DVA, to decouple the service placement, refreshing, and offloading decisions in the spatial domain based on the estimation of future transmission cost. The worst performance of DVA is proved to be bounded. Real-trace simulations and 5G testbed experiments show that DVA can reduce the backhaul transmission cost by up to 59.1\% compared with the state-of-the-art baselines. For future work, we will consider other stochastic refreshing mechanisms and dependable service placement problem. Besides, it is also an interesting direction to jointly design the service placement, computation offloading, and container life-cycle management strategies.


\end{document}